\documentclass
[journalcls]
{IEEEtran}

\label{packages}

\usepackage{amsfonts}
\usepackage{amsmath}
\usepackage{mathrsfs}
\usepackage{mathtools}
\usepackage{pifont}
\usepackage{verbatim}
\usepackage{upgreek}
\usepackage{multirow}

\usepackage{color}

\pdfoutput=1
\usepackage{enumitem}
\usepackage[usenames,dvipsnames]{xcolor}
\usepackage{caption}
\usepackage{dblfloatfix} 
\usepackage{algorithmic}
\usepackage{algorithm}
\usepackage{mdwmath}
\usepackage{mdwtab}
\usepackage{amssymb,amsmath}
\usepackage{amsmath}
\usepackage{amsthm}
\usepackage[final]{graphicx}
\graphicspath{{Images//}}
\usepackage{epstopdf}

\usepackage{cite}
\usepackage{tikz}
\usepackage{scalefnt}
\usepackage{subcaption}
\usepackage{hyperref}
\hypersetup{colorlinks=true,linkcolor=blue}

\usepackage{lipsum}

\expandafter\def\expandafter\normalsize\expandafter{%
	\normalsize
	\setlength\abovedisplayskip{2.7pt}
	\setlength\belowdisplayskip{2.7pt}
	\setlength\abovedisplayshortskip{2.7pt}
	\setlength\belowdisplayshortskip{2.7pt}
}

\label{newcommand}

\label{English Chars}
\newcommand{\bA}{\mbox{\boldmath{$A$}}}

\newcommand{\bB}{\mbox{\boldmath{$B$}}}

\newcommand{\be}{\mbox{\boldmath{$e$}}}
\newcommand{\bF}{\mbox{\boldmath{$F$}}}

\newcommand{\bI}{\mbox{\boldmath{$I$}}}

\newcommand{\bL}{\mbox{\boldmath{$L$}}}

\newcommand{\bQ}{\mbox{\boldmath{$Q$}}}

\newcommand{\bR}{\mbox{\boldmath{$R$}}}

\newcommand{\bT}{\mbox{\boldmath{$T$}}}

\newcommand{\bU}{\mbox{\boldmath{$U$}}}

\newcommand{\bX}{\mbox{\boldmath{$X$}}}
\newcommand{\bx}{\mbox{\boldmath{$x$}}}
\newcommand{\bY}{\mbox{\boldmath{$Y$}}}

\newcommand{\tbd}{\parallel}

\newcommand{\vpc}{\vspace{1pc}}
\newcommand{\mvpc}{\vspace{-1pc}}
\newcommand{\hvpc}{\vspace{-.5pc}}

\label{Roman Chars}

\newcommand{\bSigma}{\mbox{\boldmath{$\Sigma$}}}

\newcommand{\bOmega}{\mbox{\boldmath{$\Omega$}}}

\label{Theorems}
\theoremstyle{definition}
\newtheorem{theorem}{Theorem}[section]

\newtheorem{lemma}[theorem]{Lemma}

\theoremstyle{remark}
\newtheorem*{rem}{Remark}

\begin{document}
	
	\label{title}
	\title{Designing Sequence with Minimum PSL Using Chebyshev Distance and its Application for Chaotic MIMO Radar Waveform Design}
	
	\author{
		Hamid Esmaeili Najafabadi \IEEEmembership{Student Member, IEEE}\thanks{Copyright (c) 2015 IEEE. Personal use of this material is permitted. However, permission to use this material for any other purposes must be obtained from the IEEE by sending a request to pubs-permissions@ieee.org.
			
			DOI: \href{http://ieeexplore.ieee.org/document/7707413/citations}{10.1109/TSP.2016.2621728}
			
			Hamid Esmaeili Najafabadi is currently a Ph.D. student at the Department of Electrical Engineering, University of Isfahan, Hezar jarib St., Isfahan 8174657811, Iran.  Phone:  +98-3142636343; Fax: +98-3137933071; Email: Hamid.esmaeili@gmail.com.}
		Mohammad Ataei, and Mohamad F. Sabahi
		\IEEEmembership{ Member, IEEE}\thanks{Mohammad Ataei and Mohamad Farzan Sabahi are with Department of Electrical Engineering, University of Isfahan, Hezar jarib St., Isfahan 8174657811, Iran.  Phone:  +98-3137934068; Fax: +98-3137933071; Email: Ataei@eng.ui.ac.ir, sabahi@eng.ui.ac.ir.}\hvpc
	}

	\maketitle
	\IEEEpeerreviewmaketitle

	
	\begin{abstract}
		Controlling peak side-lobe level (PSL) is of great importance in high-resolution applications of multiple input multiple output (MIMO) radars.
		In this paper, designing sequences with good autocorrelation properties are studied. The PSL of the autocorrelation is regarded as the main merit and is optimized through newly introduced cyclic algorithms, namely; PSL Minimization Quadratic Approach (PMQA), PSL Minimization Algorithm, the smallest Rectangular (PMAR) and PSL Optimization Cyclic Algorithm (POCA). It is revealed that minimizing PSL results in better sequences in terms of autocorrelation side-lobes when compared with traditional integrated side-lobe level (ISL) minimization.
		In order to improve the performance of these algorithms, fast-randomized Singular Value Decomposition (SVD) is utilized.
		To achieve waveform design for MIMO radars, this algorithm is applied to the waveform generated from a modified Bernoulli chaotic system.
		The numerical experiments confirm the superiority of the newly developed algorithms compared to high-performance algorithms in mono-static and MIMO radars.
	\end{abstract}
	\begin{IEEEkeywords}
		MIMO radar waveform design; Peak side-lobe level; Chebyshev Distance.
	\end{IEEEkeywords}
	
	\section{Introduction}
	\IEEEPARstart{W}{aveform}  design is a traditional problem in radar and communication systems.
	Classically, in these systems, a matched filter is applied to detect the target or message in the presence of a background white Gaussian noise (WGN).
	In this context, a well-transmitted waveform should have low autocorrelation side-lobes for preventing false results in detection.
	Study of sequences with good autocorrelation properties, with respect to radar applications in mind, is a classic topic \cite{Frank1962,Xiong2008,Turyn1961,Golomb1965,Dai2008,Borwein2008,Boehmer1967,Baden2011}. In fact, the related literature covers a wide array ranging from bi-phase and poly-phase to more recent chaotic and algorithmic methods. For the bi-phase Barker \cite{Barker1953}, for poly-phase Golomb \cite{Golomb1965}, Frank \cite{FrankZadoffHeimiller1962} and Chu \cite{Chu1972}, for chaotic Lorenz \cite{He2005} and for algorithmic method cyclic algorithm (CA) \cite{Dai2008} and majorization minimization (MM) \cite{SongBabuPalomar2016} methods are just a few to be enumerated.
	In all the aforementioned references, a good sequence is the one with impulse-like autocorrelation.
	In general, there exist two major merits to measure the resemblance of a sequence with impulse: integrated side-lobe level (ISL) and peak side-lobe level (PSL).
	Minimizing the first was the topic of several recent publications \cite{Stoica2009,Song2015,Hao2009,SongBabuPalomar2016}.
	However, much less effort is made on minimizing the second.
	On the other hand, generating a set of sequences with minimized PSL is of great importance in high-resolution applications of MIMO radars \cite{Haimovich2008}.
	In fact, for a radar engineer autocorrelation corresponds to the output of the matched filter of the radar system. Generally speaking, the peak of the side-lobes corresponds to the falsely detected objects (false alarms), while high peak side-lobes result in masking of the low signature targets next to high signature targets. Hence, in order to have low false alarms, the peak of the side-lobes should be lowered as low as possible. In this regard, the main contribution of this article is to address the problem of PSL minimization of a sequence and acclimatizing it to MIMO radars through chaotic waveforms as the initial sequences. 
	
	In \cite{Stoica2009}, authors have solved the problem of ISL minimization for an unimodular sequence (i.e. all elements have unit absolute value).
	They derived several cyclic algorithms: CAN, CAD, CAP, and WeCAN.
	In \cite{He2009}, by generalizing some methods given in \cite{Stoica2009}, several algorithms are developed to generate waveforms appropriate in MIMO radars applications.
	In addition to autocorrelation, they minimized the cross-correlation between the generated waveforms.
	In \cite{Zhuang2011}, the waveform design in presence of clutter and white Gaussian noise is assessed.
	Accordingly, a cyclic algorithm  is developed to maximize signal-to-clutter-plus-noise ratio (SCNR) under the constant modulus constraint.
	In \cite{Cui2014}, through cyclic optimization, a computationally attractive algorithm is introduced for the synthesis of constant modulus transmit signals with good auto-correlation properties for MIMO radars.
	In a recent work \cite{Song2015}, the problem of minimizing $l_{p}$ (for $2\le p<\infty$) norm on the side-lobes of autocorrelation is addressed.
	They majorized this problem by an $l_{2}$ problem, and by solving it, they minimized the $l_{p} $-norm of the autocorrelation side-lobes.
	Then, by choosing large values for $p,$ they approached $l_{\infty }$ norm, which is indeed The PSL.
	However, in the case of peak side-lobe level, their algorithm named ``monotonic minimizer (MM) for $l_{p}$-metric", lacks the ability to suppress a specified part of the autocorrelation.
	Moreover, this minimizer actually does not minimize PSL, but minimizes an $l_{p}$-metric, then by choosing large $p$ it approaches PSL minimization.
	As noted before, like WeCAN and CAP, they lack the ability to suppress more than half of the autocorrelation side-lobes to ``almost zero". 
	
	In \cite{Jin2013}, chaos is introduced to generate phase-coded waveforms for MIMO radars.
	In \cite{Willsey2011a}, authors introduced quasi-orthogonal waveforms for wideband MIMO radars application.
	They have established that chaotic waveforms possess many desirable radar properties.
	Consequently, in the line of previous work, in \cite{Willsey2011}, the selection of parameters for Lorenz system for wideband radars is studied.
	In this article, the problem of PSL minimization is formulated and solved.
	Then, by applying chaotic waveforms as the initial sequence to the newly developed algorithms, a set of waveforms proper for MIMO radars is constructed.
	The obtained results can be summarized as follows:\vspace{-.1pc}
	\begin{itemize}
		\item  Problem formulation for PSL minimization and solving it
		
		\item  Suppression of more than half of the autocorrelation side-lobes to ``almost zero''
		
		\item  Better side-lobes both in terms of PSL and ISL
		
		\item  Ability to generate a large set of waveforms with low cross-correlation
		
		\item    Ability to deal with long sequences
	\end{itemize}\vspace{-.1pc}
	Chaotic waveforms have been and are extensively applied to radar systems \cite{Willsey2011,Ott2002,Liu2007,H.E.Najafabadi2015}.
	The chaotic systems' outputs are bounded, aperiodic and sensitive to initial conditions, resulting in low peak to average power ratio (PAPR), low autocorrelation side-lobes and low cross-correlation, respectively.
	Moreover, the simplicity of waveform generation in these systems enables the design of a set of waveforms with high cardinality, which is appropriate in MIMO radars \cite{Willsey2011a,Willsey2011}.
	In addition to these properties, chaotic waveforms are noise-like deterministic signals; therefore, 
	radars which apply these kinds of waveforms have a low probability of interception.
	
	Bernoulli map is a typical defining example for chaos \cite{Brown1998}.
	This system is one of the simplest systems with the capability of generating chaos.
	It has phase space dimensionality of two, that is, it can produce high cardinality set of waveforms of any length with very low computational effort\cite{Pollicott1998}.
	Consequently, by the juxtaposition of this chaotic system and PSL minimization algorithm, it is possible to generate an arbitrary number of sequences with low cross-correlation, low autocorrelation low-rank and low probability of interception.
	
	The organization of the article is as follows.
	The problem of PSL minimization is formulated based on Chebyshev distance and solved through some novel cyclic algorithm in section \ref{PSLMIN}.
	Then, by enhancing the time-consuming parts of the algorithm, the speed performance of this algorithm is improved in section \ref{sec:III}. Consequently, in section \ref{sec:Solv}, imposing additional constraints on the problem is studied.
	In section \ref{sec:IV}, modified Bernoulli map is introduced.
	It is shown that by applying our algorithm to the sequences generated by a chaotic system, a suitable set of waveforms for MIMO radars can be generated.
	Accordingly, in section \ref{sec:V}, the proposed methods are evaluated.
	Afterward, the conclusion is given. The last part but not the least is the appendix.
	The presented appendix is to reveal the fact that the Chebyshev distance indeed defines a complete normed vector space, known to the mathematicians as Banach space.

	\subsection{Notations}
	We use bold lower case for column vectors and bold uppercase for matrices. The symbols $\bA^{T}$ and $\bA^{H}$ represent transpose and conjugate transpose of the matrix $\bA$, respectively. $\be_{k}$ denotes the $k$-th canonical vector, a column vector with zero elements, except for the $k$-th element, which is one. $\left|a\right|$ indicates the absolute value of the scalar $a$. $\tbd.\tbd_{2}$ denote the Frobuinous norm of a vector or matrix. The norm $\tbd.\tbd_{\infty }$ for a matrix $\bA$ is defined according to $\tbd \bA\tbd_{\infty } =\max \left|A_{i,j} \right|$, that is, $\tbd \bA\tbd_{\infty }$ is the norm defined by Chebyshev distance over the space of all complex valued $N\times M$ matrices. In appendix \ref{sec:appA}, it is shown that this operator actually defines a norm on the vector space of all complex matrices. The symbol $\le ^{d}$ is the lexicographical order or dictionary order on ${\mathbb C}$. That is, $b\le ^{d} a$ if and only if,
	\begin{multline} \label{1)}  Im\{ b\} \le Im\{ a\} \, \, \, \, \vee    \\(Im\{ b\} \le Im\{ a\} \wedge Re\{ b\} =Re\{ a\} ),  \end{multline}
	where, $ \vee $ and $ \wedge $ are ``logical or" and ``logical and", respectively.
	Accordingly, by $\max ^{d}$ and $\min ^{d}$, the maximum and minimum under the dictionary order are meant. Note that, $\le ^{d}$ is a total order on ${\mathbb C}$.
\mvpc
	\section{Peak Side-lobe Level Minimization}\label{PSLMIN}
	Let $\{ x_{n} \} _{n=1}^{N}$ represent the sequence to be designed. The autocorrelation function of this sequence is:
	\begin{equation} \label{2)}
	r_{k} =\mathop{\sum }\limits_{n=k+1}^{N} x_{n} x_{n-k}^{*} =r_{-k}^{*} ,k=0,...,N-1.
	\end{equation}
	Commonly, two major merits are considered for sequence performance in the radar; the ISL and PSL. The first is defined by,
	\begin{equation} \label{3)}
	{\rm ISL\; }=\mathop{\sum }\limits_{k=1}^{N-1} |r_{k} |^{2} .
	\end{equation}
	The PSL is computed by:
	\begin{equation}
	{\rm PSL\; }=\mathop{\max }\limits_{k\ne 0} |r_{k} |.
	\label{4)}\end{equation}
	The matrix $\bX$ for the sequence $\{ x_{n} \} _{n=1}^{N}$, is defined according to:
	\begin{equation} \label{ZEqnNum290740} {\it \bX}=\begin{bmatrix}x_{1} & 0 & \cdots & 0\\
	\vdots & \ddots & \ddots & \vdots\\
	\vdots & \vdots & \ddots & 0\\
	x_{N} & \vdots & \vdots & x_{1}\\
	0 & \ddots & \vdots & \vdots\\
	\vdots & \ddots & \ddots & \vdots\\
	0 & \cdots & 0 & x_{N}
	\end{bmatrix}_{(2N-1)\times N} \end{equation}
	
	Having defined $\bX$, the autocorrelation of the sequence $\{ x_{n} \} _{n=1}^{N}$ is represented by,
	\begin{equation} \label{6)} {\it \bX}^{H} {\it \bX}=\left[\begin{array}{cccc} {r_{0} } & {r^{*} _{1} } & {\ldots } & {r_{N-1}^{*} } \\ {r_{1} } & {r_{0} } & {\ddots } & {\vdots } \\ {\vdots } & {\ddots } & {\ddots } & {r^{*} _{1} } \\ {r_{N-1} } & {\ldots } & {r_{1} } & {r_{0} } \end{array}\right].     
	\end{equation}
	Considering the fact that a good autocorrelation is the one with $r_{0} > 0$ and $r_{i} =0,i\ne 0$, a sequence could be designed, where $\bX^{H}  \bX \cong N  \bI$. Accordingly, let
	\begin{equation} \label{7)} \tbd \bA \tbd_{\infty } =\mathop{\max }\limits_{_{i,j} } \left|A_{i,j} \right|        \end{equation}
	be the norm defined by Chebyshev distance for matrix $\bA$ (refer to appendix for the proof that it is actually a  norm). Minimizing the PSL of the sequence $\{ x_{n} \} _{n=1}^{N}$ is equivalent to minimizing:
	\begin{equation} \label{ZEqnNum744207}
	\tbd\bX^{H}  \bX-N \bI \tbd_{\infty }.
	\end{equation}
	It is assumed that coping this quadratic problem is unapproachable, therefore, with an argument to be followed, minimizing the following equation instead is considered,
	\begin{equation} \label{ZEqnNum907886} \tbd \bX -\sqrt{N} \bL \tbd_{\infty }  ,        \end{equation}
	where the matrix $\bL$ with the dimensionality of $\left(2N-1\right)\times N$ satisfies $\bL^{H} \bL= \bI$, and $\bX$ is of the form defined in \eqref{ZEqnNum290740}, indicating that the problem of minimizing PSL is equivalent to the following minimization problem,
	
	\begin{equation} \label{ZEqnNum681717} \begin{array}{l} {{\rm min}\tbd\bX-\sqrt{N}  \bL\tbd_{\infty }} \\ {s.t.\, \, \bL^{H} \bL= \bI} \end{array} \end{equation}
	
	\textbf{Argument: }
	The normed space $({\mathbb C}^{N\times M} ,\tbd.\tbd_{\infty })$ is a Banach space (for the proof refer to Appendix A) and topologically equivalent to the Euclidean space $( {\mathbb C}^{N\times M} ,\tbd.\tbd_{2} )$.
	Hence, the norm is a continuous function and space is complete.
	Therefore, although \eqref{ZEqnNum744207} and \eqref{ZEqnNum907886} are not equivalent, they are almost equivalent. That is, \eqref{ZEqnNum907886} is zero if and only if \eqref{ZEqnNum744207} is zero. Hence, from the continuity, if the global minimum of \eqref{ZEqnNum907886} is sufficiently small, then the sequence where \eqref{ZEqnNum907886} is minimized, is arbitrary close to the solution of \eqref{ZEqnNum744207} $\blacksquare$.
	
	Before trying to find the solution of \eqref{ZEqnNum907886}, consider the situation where not all $\{ r_{k} \} _{k=0}^{N-1}$ but some part of them is required to be made small, as considering $r_{k } ,{\rm \; }k=1,\ldots Q-1$, where $Q \le N$.
	That is just $Q-1$ first autocorrelation coefficients are of importance.
	Accordingly, Eq. \eqref{ZEqnNum290740} should be altered properly, that is, let
	\begin{equation} \label{11)} \tilde{\bX}=\begin{bmatrix}x_{1} & 0 & \cdots & 0\\
	\vdots & \ddots & \ddots & \vdots\\
	\vdots & \vdots & \ddots & 0\\
	\vdots & \vdots & \vdots & x_{1}\\
	x_{N} & \vdots & \vdots & \vdots\\
	0 & \ddots & \vdots & \vdots\\
	\vdots & \ddots & \ddots & \vdots\\
	0 & \cdots & 0 & x_{N}
	\end{bmatrix}_{\left(N+Q-1\right)\times Q}. \end{equation}
	Then we have,
	\begin{equation} \label{6)}  \tilde{\bX}^{H} \tilde{\bX}=\left[\begin{array}{cccc} {r_{0} } & {r^{*} _{1} } & {\ldots } & {r_{Q-1}^{*} } \\ {r_{1} } & {r_{0} } & {\ddots } & {\vdots } \\ {\vdots } & {\ddots } & {\ddots } & {r^{*} _{1} } \\ {r_{Q-1} } & {\ldots } & {r_{1} } & {r_{0} } \end{array}\right].      \end{equation}

	Indeed, by considering this configuration, the PSL of selected autocorrelation coefficients $r_{k } ,k=1,\ldots Q-1,$ can be minimized through solving the following minimization problem,
	\begin{equation} \label{ZEqnNum540931}
	\begin{array}{l} {{\rm min}\tbd \tilde{\bX}-\sqrt{N} \bL \tbd_{\infty } } \\ {s.t.\, \, \bL^{H} \bL=\bI} \end{array}
	\end{equation}
	where, $\bL$  is a $\left(N+Q -1\right)\times Q$ matrix. Like any other cyclic algorithm, finding the solutions of $\tilde{\bX}$ and $\bL$ from each other in a cyclical manner would suffice. However, because of the infinity norm, it is not possible to find the $\bL$. To mitigate this problem, here, just an approximate solution of $\bL$ is sufficient. In fact, because of the cyclic nature of this algorithm, it is sufficient to have an enhancement at each step. A proper approximate solution of \eqref{ZEqnNum540931} with respect to $\bL$ can be obtained by singular value decomposition, accordingly, first let
	\begin{equation} \label{ZEqnNum170766}
	\tilde{\bX}^{H} =\bU_{1} \bSigma \bU_{2}^{H}
	\end{equation}
	be the economy-sized singular value decomposition (SVD) of the conjugate transpose of the matrix $\tilde{\bX}$, then,
	\begin{equation} \label{16)}
	\bL=\bU_{2} \bU_{1}^{H} ,
	\end{equation}
	is the solution of the Euclidean version of the problem in \eqref{ZEqnNum540931} (see \cite{Stoica2009,Li2008,He2009} and the references therein).
	\textbf{Fig.} \ref{fig1} illustrates the conceptual solution of Euclidean and Chebyshev norm version of the problem in \eqref{ZEqnNum540931}, for two dimensions.
	In this figure, the 8 points for $\tilde{\bX}$ and their projection on the feasible set $\bL^{H} \bL=\bI$ under the Euclidean and Chebyshev distance are illustrated.
	The length of the dotted and solid lines, which connect each point and its projection on the circle, describes the Chebyshev and Euclidean distance in this figure, respectively.
	Almost all other configurations of the $\tilde{\bX}$ are symmetrically equivalent to these 8 points.
	It is deduced from Fig. \ref{fig1} that the error in this approximation is less than 1/8 of the circle in all configurations.
	In this figure, the solutions of the Chebyshev distance problem (i.e. \eqref{ZEqnNum540931}), are obtained such that by relocating the solution on the circle, the triangle side with the maximum length can not be shortened.
	Note that, in 4 of the 8 configurations, the approximation errors are zero. Besides, as the dimensionality of the problem increases the approximation error occupies less proportion of the corresponding hyper-ball.
	\begin{figure}
		\raggedleft\mvpc
		\includegraphics[width=100mm,height = 80mm,page = 1]{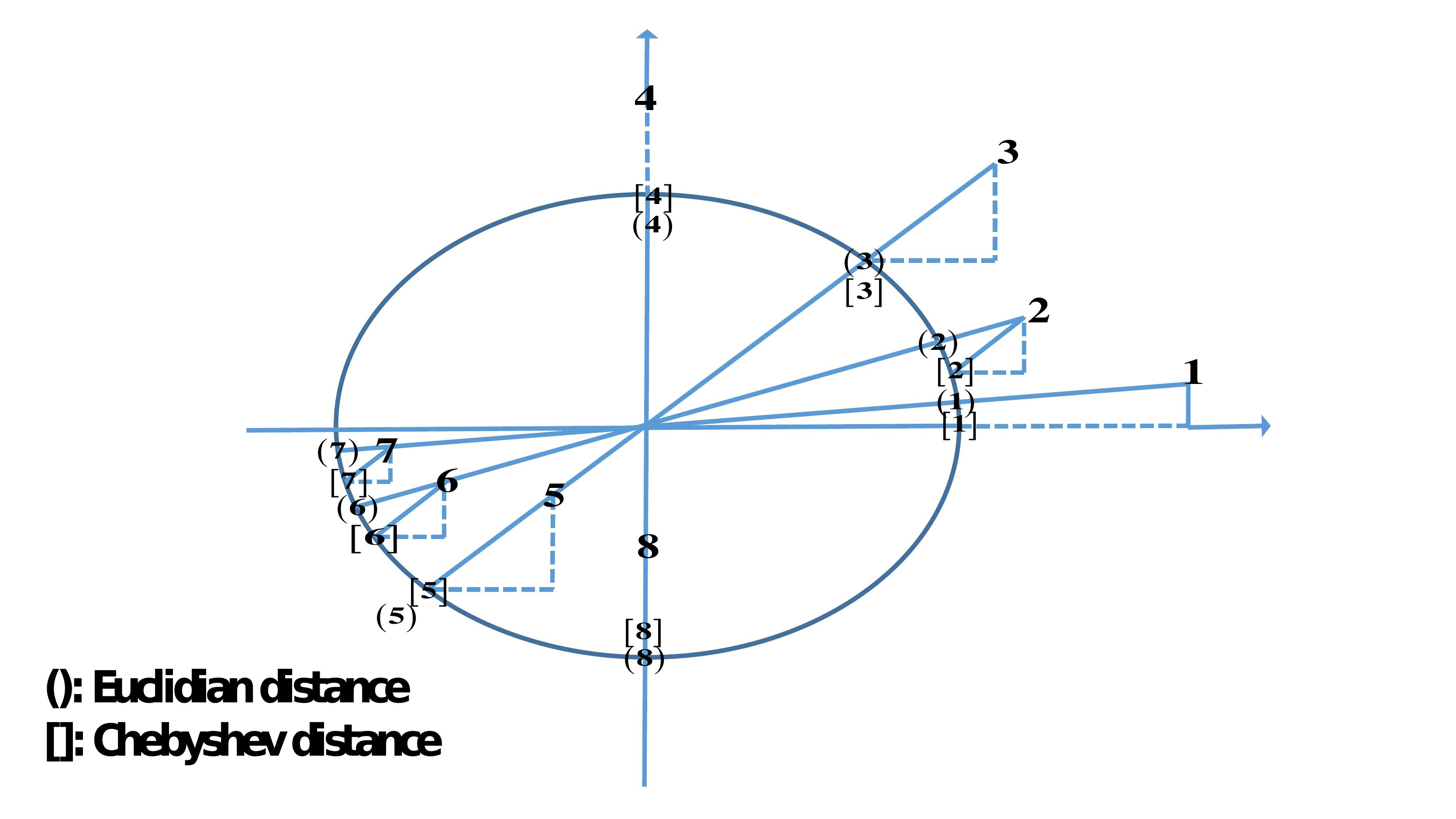}
		\caption{\small Conceptual illustration for comparison of the solution of \eqref{ZEqnNum540931} with the solution of Euclidean version of it in 2 dimension}\label{fig1}
	\mvpc\hvpc
	\end{figure}
	The second step in the cycle is to solve \eqref{ZEqnNum540931} with respect to $\{ x_{n} \} _{n=1}^{N}$ for a given $\bL$. In order to accomplish this let $x$ be an arbitrary element of $\{ x_{n} \} _{n=1}^{N}$ and $\left\{\mu _{k} \right\}$ be the corresponding elements of the matrix $\sqrt{N} \bL$ the positions of which are the same as $x$ in $\tilde{\bX}$. Then, the generic form of minimization in \eqref{ZEqnNum540931} becomes,
	\begin{equation} \label{ZEqnNum175528}
	\mathop{\min }\limits_{x} \mathop{\max }\limits_{k} \left|x-\mu _{k} \right|.
	\end{equation}
	This problem is a well-acknowledged problem known as the smallest enclosing circle problem, which has matured algorithms in order to find exact solution (for instance refer to \cite{Xu2003}). First,  note that the problem in  \eqref{ZEqnNum175528} can be converted to,
	\begin{equation}\label{18)}
	\begin{gathered}
	\mathop{\min }\limits_{x_{R} ,\, x_{I} } R  \\
	s.\, t.\, \, \sqrt{\left(x_{R} -Re\left\{\mu _{k} \right\}\right)^{2} +\left(x_{I} -Im\left\{\mu _{k} \right\}\right)^{2} } \le R, \\
	k = 1,...,Q-1,
	\end{gathered}
	\end{equation}
	where, $x_{R}$ and $x_{I}$ are the real and imaginary parts of $x$. Consider the following problem for a fixed $R\ge 0$,
	\begin{equation}\label{19)}
	\begin{gathered}
	\min \, \, \theta  \\
	s.t.\, \, \left(x_{R} -Re\left\{\mu _{k} \right\}\right)^{2} +\left(x_{I} -Im\left\{\mu _{k} \right\}\right)^{2} -\theta \le R^{2}\\
	k = 1,...,Q-1.
	\end{gathered}
	\end{equation}
	Then, the problem in \eqref{19)} becomes equivalent to the problem in \eqref{18)} in a sense that $\left(x_{R}^{*} ,x_{I}^{*} ,R^{*} \right)$ is an optimal solution of  \eqref{18)} if and only if $\left(x_{R}^{*} ,x_{I}^{*} ,0\right)$ is an optimal solution of \eqref{19)} for $R^{*} =R$ (see the proof of the theorem 2 in \cite{Xu2003}).
	Having established this equivalency, the problem in \eqref{19)} can be converted into a series of constrained quadratic programming problems,
	accordingly, define $z:=x_{R}^{2} +x_{I}^{2} -\theta .$
	Then, \eqref{19)} can be reformulated as
	\begin{equation}\label{ZEqnNum479625}
	\begin{gathered}
	\min \, \, \, x_{R}^{2} +x_{I}^{2} -z  \\
	{s.\, \, t.\, \, \, \, -2x_{R} Re\left\{\mu _{k} \right\}\, \, -2x_{I} Im\left\{\mu _{k} \right\}+z} \\
	\le R^{2} -\left(Re\left\{\mu _{k} \right\}\, \right)^{2} -\left(Im\left\{\mu _{k} \right\}\right)^{2} ,\, \, k=1...Q-1
	\end{gathered}
	\end{equation}
	This problem is a quadratic programming with linear constraints, where, if its optimal value is zero, then $R$ is the optimal value for the problem in \eqref{18)}, otherwise, reduce $R$ and solve the newly obtained quadratic problem. Accordingly, one may follow the following steps in order to get the exact solution of \eqref{18)}.
	
	Sub-algorithm 1:
	\begin{enumerate}
		\item  Start from an appropriate point e.g.,
		\begin{equation}
		x=\; \frac{\max ^{d} \left\{\mu _{k} \right\}+\min ^{d} \left\{\mu _{k} \right\}}{2} .
		\end{equation}
		and compute,
		\begin{equation} \label{eq:R}
		R=\mathop{\max }\limits_{k} \sqrt{(x_{R} -Re\{ \mu _{k} \} )^{2} +(x_{I} -Im\{ \mu _{k} \} )^{2} }
		\end{equation}
		
		\item  Solve \eqref{ZEqnNum479625} and find $z$, $x_{R}$ and $x_{I} $.

		\item  If $\left|\, \, x_{R}^{2} +x_{I}^{2} -z\right|<\delta $ then stop. Else, re-compute $R$ in \eqref{eq:R} and go to step 2.
	\end{enumerate}
	Consequently, the PSL minimization algorithm based on quadratic approach is introduced through Algorithm 1.
	In addition to this approach, other methodologies for solving the smallest circle problem in \eqref{ZEqnNum175528} can be adopted. Two other such methodologies for exact solution of the problem are introduced in appendix B.  Nevertheless, most of these solutions are computationally very expensive for our application. For instance, the cone optimization approach requires $O(Q^{3.5} \left|{\rm \; log\;}\delta \right|)$ arithmetic operations to find the optimal solution of \eqref{ZEqnNum175528}, where $\delta$ is a user specified parameter for accuracy. Similarly, in the best case scenario, the Sub-algorithm 1 requires at least $O(Q^{2})$ operations. Since the solution of \eqref{ZEqnNum175528} is required in ``every iteration", the preference is to find an approximate but fast solution.
	 \begin{tabular}{|p{0.2in}|p{2.85in}|} \hline
		
		\multicolumn{2}{|p{2.85in}|}{\textbf{Algorithm }1\textbf{:} \textbf{PMQA }(PSL Minimization Quadratic Approach)} \\ \hline
		1 & Set $\{ x_{n} \} _{n=1}^{N} $ to an initial sequence. (Initialization) \\ \hline
		2 & Constitute the matrix $\tilde{\bX}$ and compute $\bL$ according to equation \eqref{ZEqnNum170766} \\ \hline
		3 & For each $x_{n} $ constitute the sequence$\{ \mu _{k} \} _{k=1}^{Q-1} $, solve the smallest circle problem in \eqref{ZEqnNum175528} by following steps 1,2,3 as in Sub-algorithm  1  \\ \hline
		4 & Go to step 2 till some stop criterion is satisfied \newline (e.g. $\tbd \tilde{\bX}_{new} -\tilde{\bX}_{old} \tbd_{\infty } <\varepsilon $) \\ \hline
	\end{tabular}

	\section{Acceleration Schemes}\label{sec:III}
	\begin{figure}\hvpc
		\centering
		\includegraphics[trim = 0mm 30mm 0mm 30mm, clip ,width =65mm, height = 72mm]{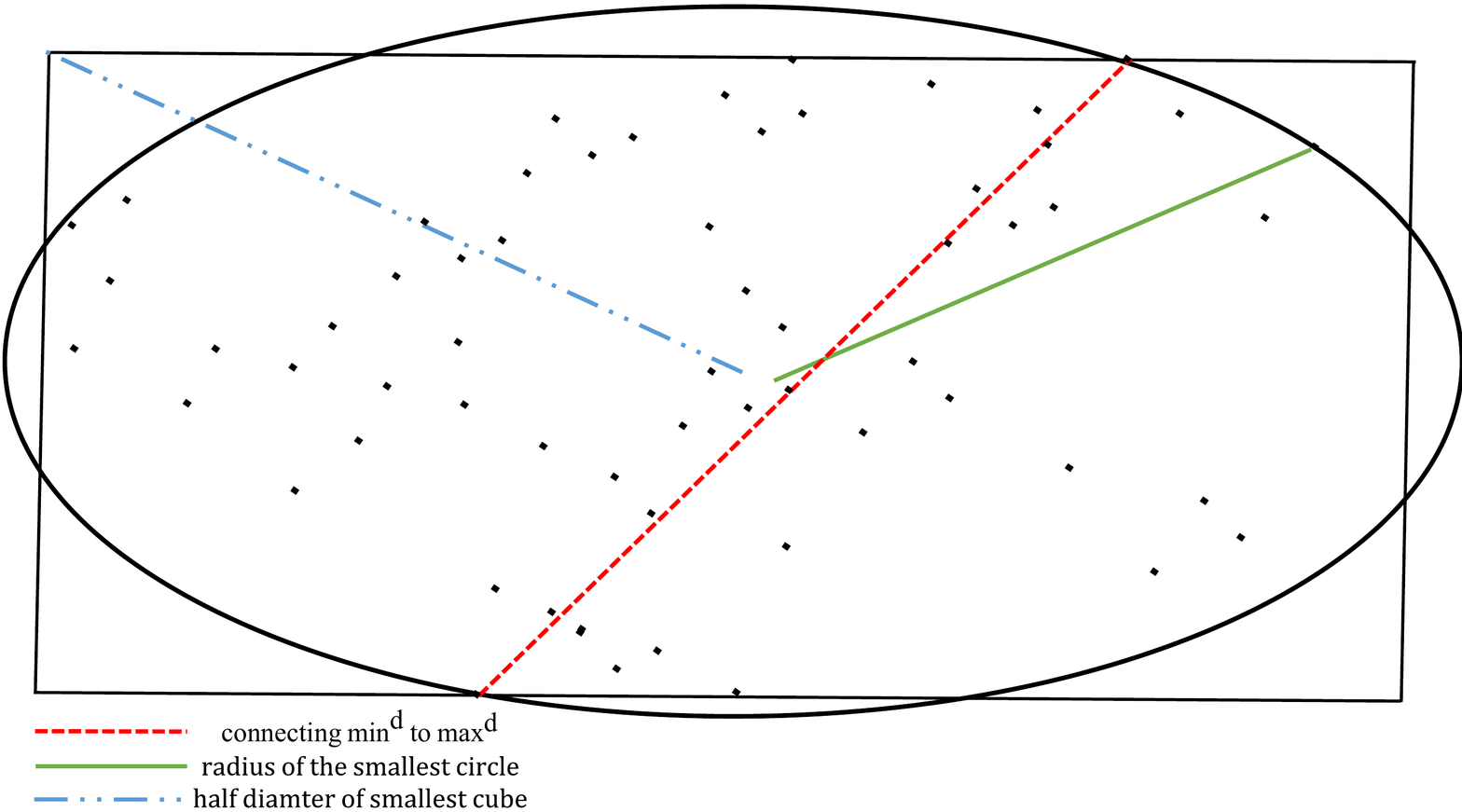}
		\caption{\small Conceptual illustration for comparing solutions to the smallest circle, smallest rectangular problems and the one introduced in \eqref{ZEqnNum457011}}
		\label{fig2}
		\vspace{-1.5pc}
	\end{figure}
	
	In order to accelerate the algorithm, the first step is to avoid the iterative Sub-algorithm 1 by some wise guess. Accordingly, instead of solving the problem in \eqref{ZEqnNum175528}, the following approximated problem is considered
	\begin{equation} \label{21)}
	\mathop{\min }\limits_{x} \mathop{\max }\limits_{k} \left|x-\mu _{k} \right|_{\infty },
	\end{equation}
	wherein  $\left|a \right|_{\infty } :=\mathop{\max }\{Re(a),Im(a)\}$.
	In fact, given the sequence $\left\{\mu _{k} \right\}$, the above mentioned problem is the smallest enclosing rectangular for this sequence.
	\textbf{Fig.} \ref{fig2}  depicts the smallest rectangular and smallest circle problems for a random sequence. In fact, in almost all cases the solution to the smallest rectangular is a good approximation to the smallest circle problem. Moreover, the smallest rectangular has a closed-form solution in the form,
	\begin{equation} \label{22)} \begin{array}{l} {x_{*} =\; \frac{\max \left\{Re\{ \mu _{k} \} \right\}+\min \left\{Re\{ \mu _{k} \} \right\}}{2} } \\ {\, \, \, \, \, +j\frac{\max \left\{Im\{ \mu _{k} \} \right\}+\min \left\{Im\{ \mu _{k} \} \right\}}{2} .} \end{array} \end{equation}
	This solution consumes $O(Q)$ logic and $O(1)$ arithmetic operations (additions and multiplications) in each iteration, hence saving a considerable amount of time. Furthermore, as a  rougher and more fast approach one may consider the following approximate solution to \eqref{ZEqnNum175528},
	\begin{equation} \label{ZEqnNum457011} x_{*} =\; \frac{\max ^{d} \left\{\mu _{k} \right\}+\min ^{d} \left\{\mu _{k} \right\}}{2} . \end{equation}
	It is worth mentioning that if one confines problem in \eqref{ZEqnNum175528} to real sequences then the exact solution of \eqref{ZEqnNum175528}  will be of the similar form (see  end of appendix B \eqref{sec:appB} for proof.),
	\begin{equation} \label{24)} x_{*} =\; \frac{\max \left\{\mu _{k} \right\}+\min \left\{\mu _{k} \right\}}{2} . \end{equation}
	Consequently, the following two algorithms are obtained for generating sequences with good autocorrelation shapes. \vspace{.5pc}
	
	\noindent
	\begin{tabular}{|p{0.2in}|p{2.85in}|} \hline
		\multicolumn{2}{|p{2.85in}|}{\textbf{Algorithm 2:} \textbf{PMAR }(PSL Minimization Algorithm where the smallest Rectangular)} \\ \hline
		1 & Set $\{ x_{n} \} _{n=1}^{N} $ to an initial sequence. (Initialization) \\ \hline
		2 & Constitute the matrix $\tilde{\bX}$ and compute $\bL$ according to equation \eqref{ZEqnNum170766} \\ \hline
		3 & For each $x_{n} $ constitute the sequence$\{ \mu _{k} \} _{k=1}^{Q-1} $, solve the smallest rectangular problem in \eqref{21)} using \eqref{22)} to achieve the new $x_{n} $ \\ \hline
		4 & Go to step 2 till some stop criterion is met \newline (e.g. $\tbd\tilde{\bX}_{new} -\tilde{\bX}_{old} \tbd_{\infty } <\varepsilon $) \\ \hline
	\end{tabular}
	
   \vpc\noindent
	\begin{tabular}{|p{0.2in}|p{2.85in}|} \hline
		\multicolumn{2}{|p{2.85in}|}{\textbf{Algorithm 3:} \textbf{POCA }(PSL Optimization Cyclic Algorithm)\textbf{}} \\ \hline
		1 & Set $\{ x_{n} \} _{n=1}^{N} $ to an initial sequence (Initialization) \\ \hline
		2 & Constitute the matrix $\tilde{\bX}$ and compute $\bL$ according to equation \eqref{ZEqnNum170766} \\ \hline
		3 & For each $x_{n} $ constitute the sequence$\{ \mu _{k} \} _{k=1}^{Q-1} $, find the maximum and minimum of this sequence under the dictionary order and set $x_{n} $ to arithmetic mean of them  \\ \hline
		4 & Go to step 2 till some stop criterion is satisfied \newline (e.g. $\tbd\tilde{\bX}_{new} -\tilde{\bX}_{old} \tbd_{\infty } <\varepsilon $) \\ \hline
	\end{tabular}\vspace{.5pc}
	Alongside the enclosed iteration, the most computationally intensive part of our algorithms and those that minimize autocorrelation related factors like CA, CAP, CAD, is SVD. This is just the problem when large values of $N$ is of concern. Reminding aforementioned argument that even an approximate value of $\bL$ works very well, fast low rank SVD algorithms can be applied (\cite{Woolfe2008} and \cite{Tropp2009}) to improve the complexity costs of the algorithms. Specifically, let $S\ll Q \le N$ and $\bOmega$ be a Gaussian random matrix of the size $(N+Q-1)\times S$. Constitute,
	\begin{equation} \label{25)}
	\bY=\tilde{\bX}^{H} \bOmega .
	\end{equation}
	where $\bY$ is of the size $Q \times S  $. Next, let
	\begin{equation} \label{26)}
	\bY=\bQ\bR\;
	\end{equation}
	be the QR factorization of the matrix $\bY$, such that $\bQ$ be $Q \times S$ unitary matrix. Then, form matrix,
	\begin{equation} \label{27)} \bB=\bQ^{H} \tilde{\bX}^{H}.\;  \end{equation}
	Note that $\bB$ has much lower dimensions $\left(S\times (N+Q-1)  \right)$ compared to original $ \tilde{\bX}$, that is, SVD of this matrix is computationally cheaper.  Compute the SVD of this thiner matrix,
	\begin{equation} \label{28)}
	\bB=\hat{\bU}_{1} \hat{ \bSigma }\bU_{2}^{H}.
	\end{equation}
	Finally, form the orthonormal matrix $\hat{ \bU}_{1}$, and compute $ \bU_{1}$ according to
	\begin{equation} \label{29)}
	\bU_{1} =\bQ\hat{\bU}_{1}
	\end{equation}
	It is known that, if the value of $S$ is chosen large enough, then the initial matrix $ \tilde{\bX}^{H}$ can be approximated arbitrarily close by (see \cite{Woolfe2008,Tropp2009} and references therein),
	\begin{equation} \label{ZEqnNum770037}
	\tilde{\bX}^{H} \approx \bU_{1} \hat{\bSigma }\bU_{2}^{H} ,
	\end{equation}
	In fact, as observed in the aforementioned references even for small values of $S$ this works well. The experiments run here indicate that for $S$ values as small as 4 suffice.
	
	By applying this technique, those versions of CA algorithm where SVD is involved (see, e.g. \cite{Stoica2009,He2009}) can be made faster. Regarding the algorithm introduced here, i.e. the POCA, this enhancement results in the following algorithm, where it is named RPOCA for randomized PSL optimization cyclic algorithm. The other SVD based cyclic algorithms, like PMQA, PMAR, CA, CAP, and CAD can be dealt with similarly. To avoid unnecessarily lengthening of the article, they are not developed here.
	\vspace{.7pc}
	\begin{tabular}{|p{0.1in}|p{2.85in}|} \hline
		\multicolumn{2}{|p{2.85in}|}{\textbf{  Algorithm 4: RPOCA }(Randomized PSL Optimization Cyclic Algorithm)} \\ \hline
		1 & Set $\{ x_{n} \} _{n=1}^{N} $ to an initial sequence. Choose $S\ll Q \le N$  and set $\bOmega $ to a Gaussian random matrix. (Initialization) \\ \hline
		2 & Constitute the matrix $\tilde{\bX}$ (equation\eqref{11)}) and compute the low rank SVD of it using the equations \eqref{ZEqnNum770037} then form $\bL$ according to equation \eqref{ZEqnNum170766}  \\ \hline
		3 & For each $x_{n} $ constitute the sequence $\{ \mu _{k} \} _{k=1}^{Q-1} $, find lexicographic maximum and minimum of this sequence according to \eqref{ZEqnNum457011} and set $x_{n} $ to their arithmetic mean.  \\ \hline
		4 & Go to step 2 till some stop criterion is satisfied.  \\ \hline
	\end{tabular}

	Note that the computational cost of low-rank SVD is in the order of $O((N+Q -1)QS)$. Besides, RPOCA needs no random access to matrix $\tilde{\bX}$, that is, it can be implemented simpler (see \cite{Tropp2009,Woolfe2008}).
	In fact, in computing a low-rank SVD version of $\tilde{\bX}$ (equation \eqref{ZEqnNum770037}), just the matrices of the size $S{\rm \times }Q $ should be fitted in the random access memory (RAM).
	On the contrary, the conventional SVD needs all the $\tilde{\bX}$, with dimensions of $(N+Q -1){\rm \times }Q $, to be fitted in RAM.
	\mvpc
	\section{Solving the problem with constraints}\label{sec:Solv}
	No additional constraint was assigned to the developement of the above algorithms. 
	However, based on practical requirements some constraints can be added to this kind of development. Such constraints have been and are being adopted by many authors, like \cite{Stoica2009,SongBabuPalomar2016,Song2015,Hao2009}. One of such requirements is the unimodularity constraint. The sequence $ \{ x_{n} \} _{n=1}^{N} $ is called unimodular if and only if, 
	\begin{equation}\label{key}
	|x_n| = 1,    n = 1,...,N.
	\end{equation} 
	Unimodular sequences have the lowest possible peak to average power ratio (PAPR). Alternatively, one might consider restricting PAPR itself. In order to enforce such a restriction define $ \bx = [x_1,...,x_n]^{T} $. Since in our developed $\bX^{H}  \bX \cong N  \bI$ the average power is $ \frac{1}{N}\bx^{H}\bx\cong 1 $. Hence, in order to restrict PAPR it is sufficient to restrict peak power or equivalently adding the following constraints to each problem in the development of previous sections.
	\begin{equation}\label{key}
	\||\bx|\|_{\infty}\le a
	\end{equation}
	wherein $ |.| $ denotes the element-wise absolute value of the complex vector $ \bx $ and $ \|. \|_\infty $ denotes the Chebyshev norm of $ \bx $.
	In order to preserve brevity, we avoid rewriting the 
	whole problems again. Nevertheless, such constraints should be added to the equations \eqref{ZEqnNum744207}, \eqref{ZEqnNum681717}, \eqref{ZEqnNum540931}, \eqref{ZEqnNum175528}, \eqref{18)}, \eqref{19)}, \eqref{ZEqnNum479625},\eqref{eq:R}. This alteration, illustrates itself in the sub-algorithm 1, Eq. \eqref{eq:R}, where the equation becomes, 
	\begin{gather}\label{key}
	R=\mathop{\max }\limits_{k} \sqrt{(x_{R} -Re\{ \mu _{k} \} )^{2} +(x_{I} -Im\{ \mu _{k} \} )^{2} }\\\nonumber
	x_{I}^2+x_R^2 = 1
	\end{gather}
	for unimodularity constraint and 
	\begin{gather}\label{key}
	R=\mathop{\max }\limits_{k} \sqrt{(x_{R} -Re\{ \mu _{k} \} )^{2} +(x_{I} -Im\{ \mu _{k} \} )^{2} }\\\nonumber
	x_{I}^2+x_R^2 \le a
	\end{gather}
	for constraining the PAPR. It is interesting to mention that the development of acceleration schemes (Sec. \ref{sec:III}) subject to the above-mentioned constraints are also possible. Accordingly, \eqref{21)} should be solved subject to the these constraints. 
	\begin{gather} \label{}
	\mathop{\min }\limits_{x} \mathop{\max }\limits_{k} \left|x-\mu _{k} \right|_{\infty },\nonumber\\
	|x| = 1 \;\; \texttt{or} \;\;|x|\le a
	\end{gather}
	where the first constraint is for unimodularity and the second is for constraining PAPR. Regarding unimodularity constraint, finding an approximate and fast solution can be achieved by solving it without constraint and then project it to unit circle. Accordingly, the alteration on the algorithms PMAR, POCA and RPOCA can be achieved by introducing following intermediate step between steps 3 and 4.
	
	\begin{tabular}{|c|c|}
		\hline 
		3.5 & Project $x_{n}$ to unit circle.   \hspace{3.2cm}     \\
		\hline 
	\end{tabular}

	Similarly, in the case of restricted PAPR, it is sufficient to solve \eqref{21)} without any constraints and if the solution satisfies the constraint (i.e. $ |x| \le a $), mission accomplished, else the projection onto  $ |x| = a $ circle will restrict the PAPR. Therefore, the alteration for algorithms PMAR, POCA and RPOCA can be made by introducing the following intermediate step between steps 3 and 4. 
	
	\begin{tabular}{|c|c|}
		\hline 
		3.5 & If $x_{n} > a$ then project $x_{n}$ to the $ |x| = a $ circle.\\
		\hline 
	\end{tabular}
	\begin{rem}
		By introducing unimodularity constraint, a trade-off involving cross-correlation and autocorrelation are derived \cite{Welch1974}. Consider the uni-power set of sequences $ \{x_n^m\}, n = 1,..., N, m = 1,...,M$, where
		\begin{equation}\label{key}
		\sum^N_{n=1}|x_{n}^{m}|^{2}=1,m=1,...,M .
		\end{equation}
		It is known that for such a sequence the following lower-bound exists. 
		\begin{equation}\label{}
		c_{max} \ge \sqrt{\frac{M-1}{M(2N-1)-1}},
		\end{equation}
		where, $ c_{max} $ is the maximum of the correlation side-lobes (consisting of all cross-correlation and all autocorrelation lags except zero). When sequences are unimodular, each sequence will be of power $ N $, (i.e. $ \sum^N_{n=1}|x_{n}^{m}|^{2}=N,m=1,...,M  $). Therefore, the correlation scales up by a factor of $ N $. Accordingly, the lower-bound in this case is of the form, 
		\begin{equation}\label{key}
		c_{max} \ge N\sqrt{\frac{M-1}{M(2N-1)-1}}.
		\end{equation}
	\end{rem}
	Nevertheless, all the above algorithms generate sequences with good autocorrelation.
	In the next section, by applying the modified Bernoulli system, an ``arbitrarily large" set of sequences with low cross-correlation is generated. By combining these two approaches, one can design large set of sequences with low cross-correlation and arbitrary shape of autocorrelation.
	
	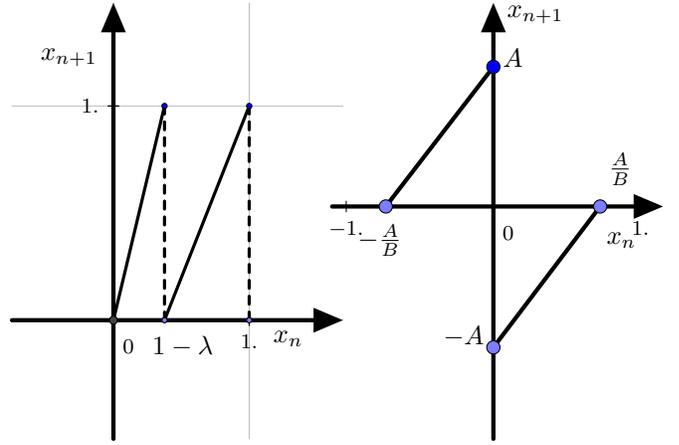
\begin{figure}\label{fig3}\hvpc
		\raggedright
		
		\definecolor{xdxdff}{rgb}{0.49019607843137253,0.49019607843137253,1.}
		\definecolor{qqqqff}{rgb}{0.,0.,1.}
		\definecolor{uuuuuu}{rgb}{0.26666666666666666,0.26666666666666666,0.26666666666666666}
		\definecolor{cqcqcq}{rgb}{0.7529411764705882,0.7529411764705882,0.7529411764705882}\usetikzlibrary{arrows}
		\begin{subfigure}[b]{0.18\textwidth}
			\begin{tikzpicture}[line cap=round,line join=round,>=triangle 45,x=1.8051679589260492cm,y=2.8496319233111502cm]
			\draw [color=cqcqcq,, xstep=1.8051679589260492cm,ystep=2.8496319233111502cm] (-0.7471455587585493,-0.5540505122049281) grid (1.6903007620911172,1.4813000720419833);
			\draw[->,ultra thick,color=black] (-0.7471455587585493,0.) -- (1.6903007620911172,0.);
			\foreach \x in {,1.}
			\draw[shift={(\x,0)},color=black] (0pt,2pt) -- (0pt,-2pt) node[below] {\footnotesize $\x$};
			\draw[->,ultra thick,color=black] (0.,-0.5540505122049281) -- (0.,1.4813000720419833);
			\foreach \y in {,1.}
			\draw[shift={(0,\y)},color=black] (2pt,0pt) -- (-2pt,0pt) node[left] {\footnotesize $\y$};
			\draw[color=black] (0pt,-10pt) node[right] {\footnotesize $0$};
			\clip(-0.7471455587585493,-0.5540505122049281) rectangle (1.6903007620911172,1.4813000720419833);
			\draw [line width=1.2pt] (0.,0.)-- (0.37624,1.);
			\draw [line width=1.2pt] (0.37624,0.)-- (1.,1.);
			\draw [line width=1.2pt,dash pattern=on 3pt off 3pt] (0.37624,1.)-- (0.37624,0.);
			\draw [line width=1.2pt,dash pattern=on 3pt off 3pt] (1.,1.)-- (1.,0.);
			\draw (-0.6033535516236147,1.3099619772725792) node[anchor=north west] {$x_{n+1}$};
			\draw (1.1109504413927537,-0.00865084210210264) node[anchor=north west] {$x_{n}$};
			\draw (0.21797106334632738,-0.02822693676928266) node[anchor=north west] {$1-\lambda$};
			\begin{scriptsize}
			\draw [fill=uuuuuu] (0.,0.) circle (1.5pt);
			\draw [fill=qqqqff] (0.37624,1.) circle (1.0pt);
			\draw [fill=xdxdff] (0.37624,0.) circle (1.0pt);
			\draw [fill=qqqqff] (1.,1.) circle (1.0pt);
			\draw [fill=xdxdff] (1.,0.) circle (1.0pt);
			\end{scriptsize}
			\end{tikzpicture}
		\end{subfigure}
		\hfil
		\begin{subfigure}[b]{0.18\textwidth}
			\definecolor{xdxdff}{rgb}{0.49019607843137253,0.49019607843137253,1.}
			\definecolor{qqqqff}{rgb}{0.,0.,1.}\usetikzlibrary{arrows}
			\begin{tikzpicture}[line cap=round,line join=round,>=triangle 45,x=1.9563943211003074cm,y=3.4480546272064925cm]
			\draw[->,ultra thick,color=black] (-1.095939250132553,0.) -- (1.1530961066688152,0.);
			\foreach \x in {-1.,1.}
			\draw[shift={(\x,0)},color=black] (0pt,2pt) -- (0pt,-2pt) node[below] {\footnotesize $\x$};
			\draw[->,ultra thick,color=black] (0.,-0.89668031103782) -- (0.,0.7854276098273958);
			\foreach \y in {}
			\draw[shift={(0,\y)},color=black] (2pt,0pt) -- (-2pt,0pt) node[left] {\footnotesize $\y$};
			\draw[color=black] (0pt,-10pt) node[right] {\footnotesize $0$};
			\clip(-1.095939250132553,-0.89668031103782) rectangle (1.1530961066688152,0.7854276098273958);
			\draw (0.03101772908895698,0.8094087818152909) node[anchor=north west] {$x_{n+1}$};
			\draw (0.7067153901115547,-0.062802829198256581) node[anchor=north west] {$x_{n}$};
			\draw [line width=1.6pt] (0.,0.53876)-- (-0.7311782344953831,0.);
			\draw [line width=1.6pt,] (0.725509634115689,0.)-- (0.,-0.5438144946445058);
			\draw (0.7189005605618267,0.24413829924347702) node[anchor=north west] {$\frac{A}{B}$};
			\draw (-0.9837314123745671,-0.03173363306497724) node[anchor=north west] {$-\frac{A}{B}$};
			\draw (0.0017461192390476347,0.6449664596125814) node[anchor=north west] {$A$};
			\draw (-0.40072718254334896,-0.4273345164175868) node[anchor=north west] {$-A$};
			\begin{scriptsize}
			\draw [fill=qqqqff] (0.,0.53876) circle (2.5pt);
			\draw [fill=xdxdff] (-0.7311782344953831,0.) circle (2.5pt);
			\draw [fill=xdxdff] (0.725509634115689,0.) circle (2.5pt);
			\draw [fill=xdxdff] (0.,-0.5438144946445058) circle (2.5pt);
			\end{scriptsize}
			\end{tikzpicture}
		\end{subfigure}

		\caption{\small Comparison of the phase space of modified Bernoulli system (right), with phase space of Bernoulli System (left)}
		\vspace*{-1\baselineskip}
		\vspace{-.5pc}
		
	\end{figure}
	\begin{figure*}\mvpc\mvpc\hvpc
		\centering
		\includegraphics[trim = 8mm 100mm 8mm 18mm, clip,width=180mm,height = 50mm]{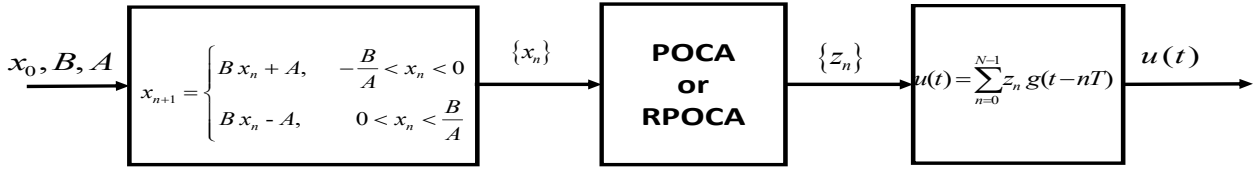}
		\label{fig4}
		\vspace*{-4\baselineskip}
		\caption{\small Baseband processing block diagram of MIMO radar at each antenna}\mvpc\hvpc
	\end{figure*}\mvpc
	
	\section{Modified Bernoulli Map}\label{sec:IV}
	Bernoulli shift map belongs to a family called piecewise linear maps, where its elements consist of a number of piecewise linear segments \cite{Ott2002,Brown1998}. Due to its chaotic nature, sequences generated from this map are highly sensitive to initial conditions.
	Therefore, by changing its initial condition, very different sequences can be produced (see \cite{Willsey2011a}). Bernoulli map is renowned for its simple structure consisting of two linear segments as follows:
	
	\begin{equation} \label{ZEqnNum605756}
	x_{n+1} =\left\{\begin{array}{c} {\begin{array}{cc} {\frac{x_{n} }{\lambda } } & {\; \; \; \; \; \; \; \; \; \; \; \; \; 0<x_{n} <} \end{array}1-\lambda } \\
	{\begin{array}{cc} {\frac{x_{n} -\left(1-\lambda \right)}{\lambda } }  \; \; \; \;& {1-\lambda <x_{n} <1} \end{array}} \end{array}\right.
	\end{equation}
	The system defined above is commonly acknowledged as Bernoulli system in its special case of $ \lambda =1/2$, represented in the form,
	\begin{equation} \label{32)} x_{n+1} =2\; x_{n} \; \, \, \, {\rm mod}\; \, \, 1,\, \, x_{0} \in \left(0,1\right). \end{equation}
	
	The classical form of Bernoulli map given in \eqref{ZEqnNum605756} is not appropriate to be applied here, since it is observed that, the signals generated from this map are all positive which result in a non-uniform high cross-correlation. To alleviate this phenomenon, the origin of the phase space of this dynamical system to is changed $\left(\frac{1}{2} ,{\it \; }1-{\it \lambda }\right)$, hence the following dynamical system, a modified version of Bernoulli system.
	\begin{equation}  \label{ZEqnNum736077}
	x_{n+1} =\left\{\begin{array}{c} {\begin{array}{cc} {Bx_{n} +A,} & {\; \; -\frac{B}{A} <x_{n} <\; } \end{array}0} \\ {\begin{array}{cc} {Bx_{n} -A,} & {\; \; \, \, \, \, \, 0<x_{n} <\frac{B}{A} } \end{array}} \end{array}\right.    \end{equation}
	\textbf{Fig. 3}  illustrates the phase space of the modified Bernoulli system along with original Bernoulli system. Note that in addition to the change in the center of the phase space, the slopes of the lines are altered as well. In \textbf{Theorem} \ref{th1}, it is revealed that the modified Bernoulli system can indeed produce chaos.\hvpc
	\begin{theorem}\label{th1}
		The modified Bernoulli system is a chaotic system for $B>1.$
	\end{theorem}
	\begin{proof}
		In order to maintain that the one-dimensional system
		\begin{equation} \label{34)}
		x_{n+1} =f\left(x_{n} \right),\; \; x\in I
		\end{equation}
		is chaotic, its being of positive Lyapunov exponent, defined by (see \cite{Collet2008} and references therein) must be proved,
		\begin{equation} \label{35)}
		\lambda =\mathop{\lim }\limits_{n\to \infty } \frac{1}{n} {\rm \; log\; |\; }\frac{d}{dx} f^{n} (x){\rm |},
		\end{equation}
		where $f^{n}(x) \triangleq\overbrace{f(f...(f}^{n\, \, times}(x)))$. In case of modified Bernoulli system by successive substitution,
		\begin{equation} \label{36)} f^{n} (x)=B^{n} x+K, \end{equation}
		is yield, where, $K$ is a constant depending solely on $B,\, A$ and $n$. Therefore,
		\begin{equation} \label{37)}
		\lambda =\mathop{\lim }\limits_{n\to \infty } \frac{1}{n} {\rm \; log\; |\; }\frac{d}{dx} f^{n} (x){\rm |}=\log |B|.
		\end{equation}
		Now, by selecting $B>1$, the system's Lyapunov exponent becomes positive.
	\end{proof}

	Chaotic waveforms possess many appropriate radar properties. In fact, the autocorrelation and cross-correlation, due to their aperiodicity and sensitivity to the initial condition, are usually low.  Besides, they can be generated with a very low computational burden for any length and quantity.

	By choosing different initial conditions, modified Bernoulli system can produce waveforms with very low cross-correlation. By applying the algorithms introduced in the previous section, the autocorrelation side-lobes levels of the waveforms can be enhanced. In order to have a proper set of sequences, it is sufficient to generate random sequences from the modified Bernoulli system beginning different initial conditions and then applying one of the algorithms introduced in previous sections (i.e. PMQA, PMAR, POCA or RPOCA). A block diagram is proposed for baseband waveform generation at the transmitter side in each antenna in \textbf{Fig. 4}.
	\hvpc
	
	\section{Simulation Results} \label{sec:V}

	\subsection{Autocorrelation Performance}
	
	In this section, the proposed methods in previous sections are evaluated. Subsequently, in order to be able to compare the autocorrelation, some metrics are defined herein. The normalized autocorrelation is defined according to,
	\begin{multline}\label{ZEqnNum228643}
	Normalized\, Autocorrelation\; \; =20\, \, {\rm log}_{10} \left|\frac{r_{k} }{r_{0} } \right|,
	\\k=1,...,N.
	\end{multline}
	This metric resembles the correlation level given in \cite{Stoica2009,Song2015}. The name ``Normalized Autocorrelation" is preferred here to avoid confusion with cross-correlation. Afterwards, the peak correlation level (PCL) is,
	\begin{equation} \label{ZEqnNum339907}
	PCL\; =20\, \, {\rm log}_{10} \left|\frac{\max r_{k} }{r_{0} } \right|.
	\end{equation}
	
	When suppressing a specified part of the autocorrelation is of concern, in \cite{Stoica2009} a useful metric named the modified merit factor (MMF) is developed,
	\begin{equation} \label{ZEqnNum569172}
	MMF=\frac{N^{2} }{2\mathop{\sum }\nolimits_{i=1}^{Q-1} |r_{{i} } |^{2} } .
	\end{equation}
	The equivalent merit factor of the suppressed part of the autocorrelation is measured through this metric. Similarly, a modified peak correlation level (MPCL) is defined according to
	\begin{equation} \label{ZEqnNum826115}
	MPCL=\left|\frac{\max\{ r_{i }|i=1,...,Q-1\} }{r_{0} } \right|.
	\end{equation}
	
	Note that MPCL is similar to PCL when suppressing all autocorrelation side-lobes is of concern. The merits defined in \eqref{ZEqnNum228643} and \eqref{ZEqnNum569172} are chosen from  \cite{Song2015,Stoica2009,He2009} in order to compare newly developed algorithms with theirs on their footings. The merits in \eqref{ZEqnNum339907} and \eqref{ZEqnNum826115} are defined here to minimize peak side-lobe level. Hereafter, we compare POCA and RPOCA with other methods. The CAN and WeCAN are selected from \cite{Stoica2009}, the ``Monotonic minimizer for Weighted ISL'' (MWISL) and the ``monotonic minimizer (MM) for $l_{p}$" are chosen from \cite{Song2015}. The CAN, WeCAN and MWISL minimize ISL, and ``MM for $l_{p} $" minimizes  $l_{p} $-norm on side-lobes. Like \cite{Song2015}, very large $p$ is considered for ``MM for $l_{p} $"  in order to approximate the Chebyshev norm, with $p$ set at 10000 in specific. For both the ``MM for $l_{p} $" and MWISL the iteration counter is set to $10^{6} $ (which is way beyond the suggestions in \cite{Song2015}) to ensure the convergence. Unless specified otherwise, all the algorithms are initialized by Golomb sequence, defined by,
	\begin{equation} \label{42)}
	G(n)=e^{(j(n-1)n\pi /N)} ,\; n=1,...,N.
	\end{equation}
	Finally, it should be mentioned that all the following simulations are implemented through MatLab on an i7 3.2 GHz machine with 6 GBytes of RAM.
	\begin{enumerate}
		\item Comparison with Barker
	\end{enumerate}
	\begin{figure}\mvpc
		\centering
		\includegraphics[width=90mm,height=.7\linewidth]{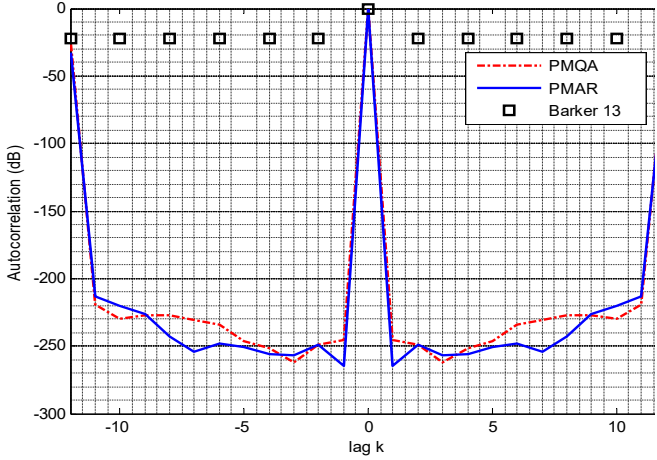}
		\caption{\small Comparison for PMQA, PMAR and Barker, all at length 13}\label{fig5}
		\mvpc\hvpc
	\end{figure}

	One of the widely embraced sequences in the context of the good autocorrelation sequences is Barker. Barker is the most acknowledged member of the bi-phase sequences.
	Although it is not fair to compare a bi-phase sequence to a polyphase sequence, due to its prevalent application, comparison with Barker is illustrative.
	The longest Barker sequence is of the length 13 and is as follows,
	\begin{equation} \label{43)}
	x = [1\; 1\; 1\; 1\; 1\; -1\; -1\; 1\; 1\; -1\; 1\; -1\; 1].
	\end{equation}
	The comparison between the autocorrelation of Barker 13, PMQA and PMAR is depicted in \textbf{Fig.} \ref{fig5} for parameters $\varepsilon =10^{-12} $, $N=13$, $Q=12.$
	The fact that both algorithms can suppress the autocorrelation side-lobes to almost zero in wide range of correlation lags is revealed by this figure.
	Here, the difference is in the consumed time, since  2.26 sec and 1.87 sec are consumed by PMQA and PMAR to generate this figure.
	In the experimentations hereafter the focus is on POCA and RPOCA, the reason is being the degree of approximations applied in deriving them.
	Accordingly, their applicability is asserted through various scenarios.
	
	\begin{figure}
		\centering
		\includegraphics[width=.51\textwidth,height = .7\linewidth]{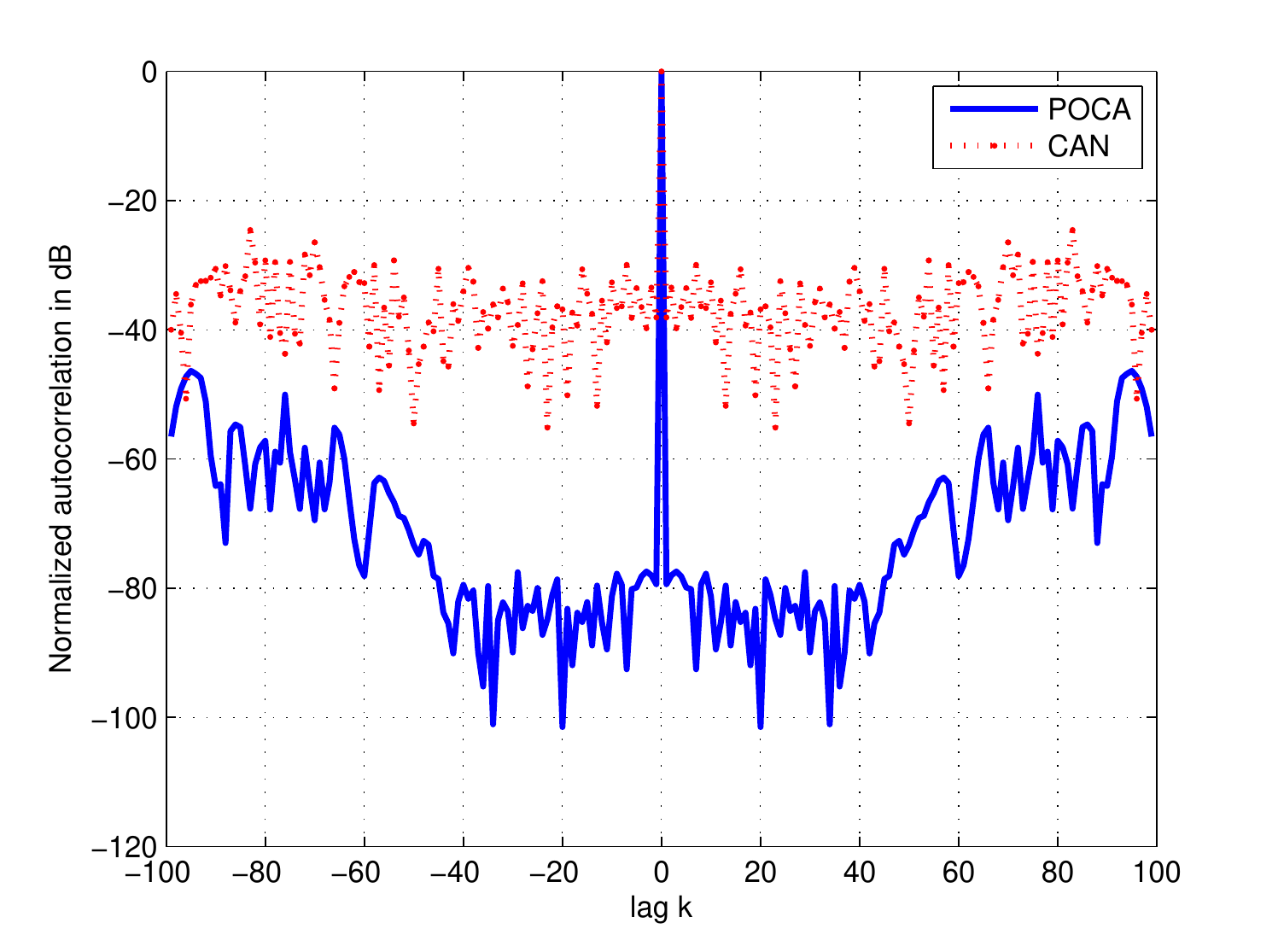}
		\caption{\small Normalized autocorrelation of POCA vs. CAN\cite{He2009}}\label{fig6}
	\mvpc\hvpc
	\end{figure}
	
	\begin{figure}
		\vspace{-.8pc}
		\centering
		\includegraphics[width=.51\textwidth,height = .7\linewidth]{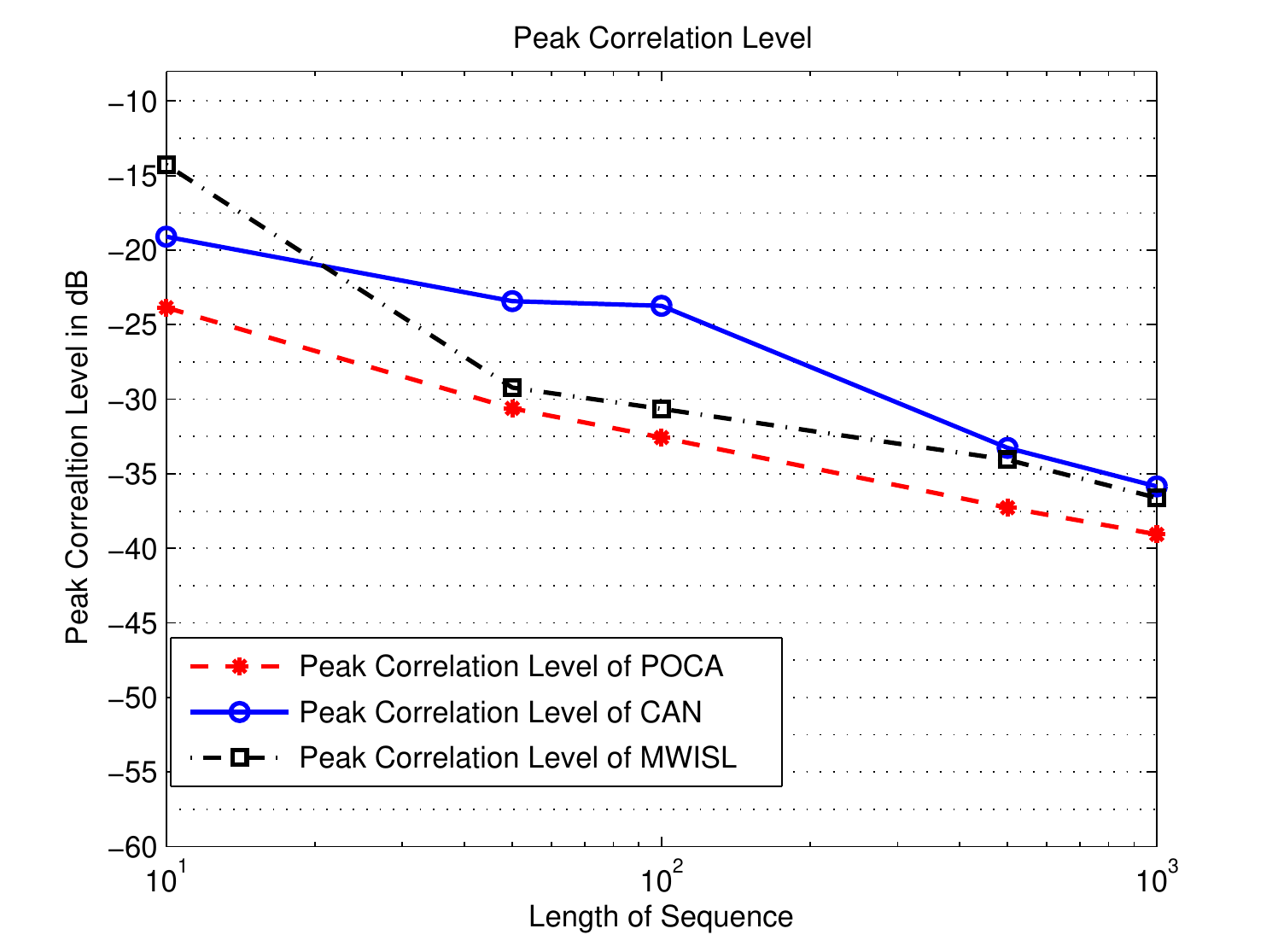}\hvpc
		\caption{\small Comparison of PCL of CAN, POCA and MWISL algorithms}\label{fig7}
		\vspace{-1pc}
	\end{figure}

	\begin{enumerate}[resume]
		\item Suppressing all side-lobes
	\end{enumerate}
	
	First, consider a scenario where suppressing all the side-lobes of the sequence is of equal importance.
	For instance, consider a case where designing a sequence with length $N=100$ is required, thus,
	\begin{equation} \label{ZEqnNum954109} \left\{\begin{array}{c} {Q=100\, \, \, } \\ { N=100 \, \, } \end{array}\right. ,
	\end{equation}
	In this scenario, $\bT=\bI_{N\times N}.$  The comparison between POCA and CAN algorithms  for $N=100$ is depicted in \textbf{Fig.} \ref{fig6}.
	The CAN algorithm is the best performing algorithm amongst several algorithms developed in \cite{Stoica2009}: CAP, CAN, and WeCAN, while, it lacks the ability to suppress some specified part of the autocorrelation, where, WeCAN is developed there to accomplish this task. Note that WeCAN and CAN have the same functionality when suppression all side-lobes is contemplated.

	The CAN, POCA and MWISL are compared in terms of the PCL in
	\textbf{Fig.} \ref{fig7} for different lengths for the same scenario as in \eqref{ZEqnNum954109}. It is observed that POCA easily beats MWISL and CAN with respect to PCL, where, the supremacy of POCA over CAN is of order 5 dB.
	This figure is generated by setting equal terminating points for both CAN and POCA.

	\begin{figure}\mvpc\hvpc
		\hspace{-1.5pc}
		\includegraphics[width=1.1\linewidth]{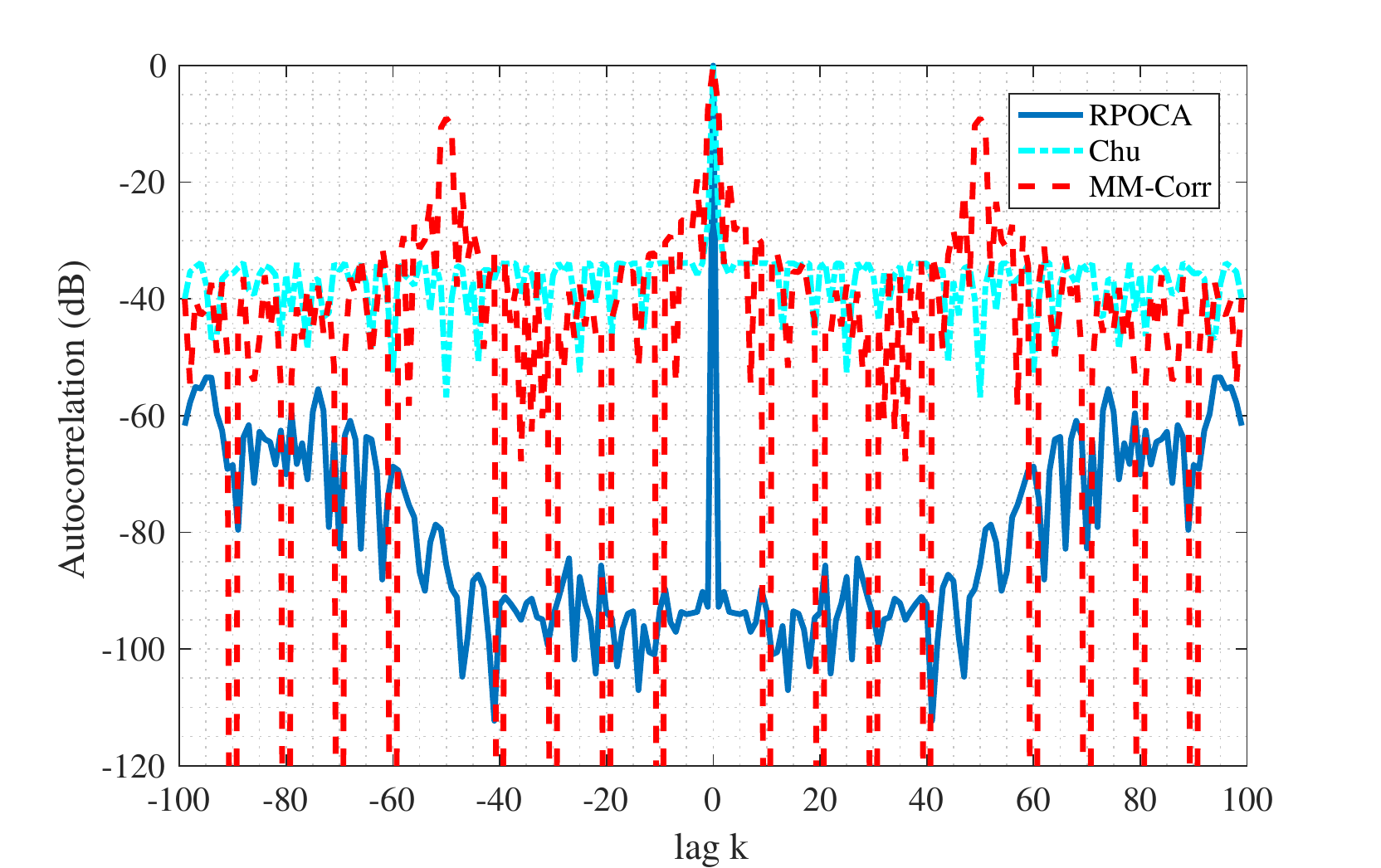}
		\caption{\small Comparison between the Normalized Autorotation of monotonic minimizer (MM)  for PSL minimization, RPOCA, and Chu sequence, for $N=100$}
		\label{fig8}
		\vspace{-1.5pc}
	\end{figure}
	
	Golomb sequence belongs to the family of polyphase sequences.
	Another member of this family is the so-called Chu sequence, which is considered to be the best amongst them ISL (refer to \cite{Mercer2013} and references therein).
	This sequence is defined according to
	\begin{equation} \label{45)} C(n)=\left\{\begin{array}{c} {e^{(j(n-1)^{2} \pi /N)} \, \, \, \, \, {\rm if}\, n\, \, {\rm is}\, \, even\, \, } \\ {e^{(j(n-1)n\pi /N)} \, \, \, \, \, {\rm if}\, n\, \, {\rm is}\, \, odd\, \, \, \, } \end{array}\right. ,n=1,...,N.
	\end{equation}
	The comparison between ``MM for $l_{p} $", RPOCA, and Chu sequence is demonstrated in \textbf{Fig.} \ref{fig8} for the same scenario as above (i.e. equation \eqref{ZEqnNum954109}),
	indicating  that, in terms of autocorrelation side-lobe levels RPOCA defeats state-of-the-art methods. 
	Like in \cite{Song2015}, the $p$ is set at 10000 in order to approximate the PSL minimization.
	\begin{figure}
		\includegraphics[width=1\linewidth,height=.6\linewidth]{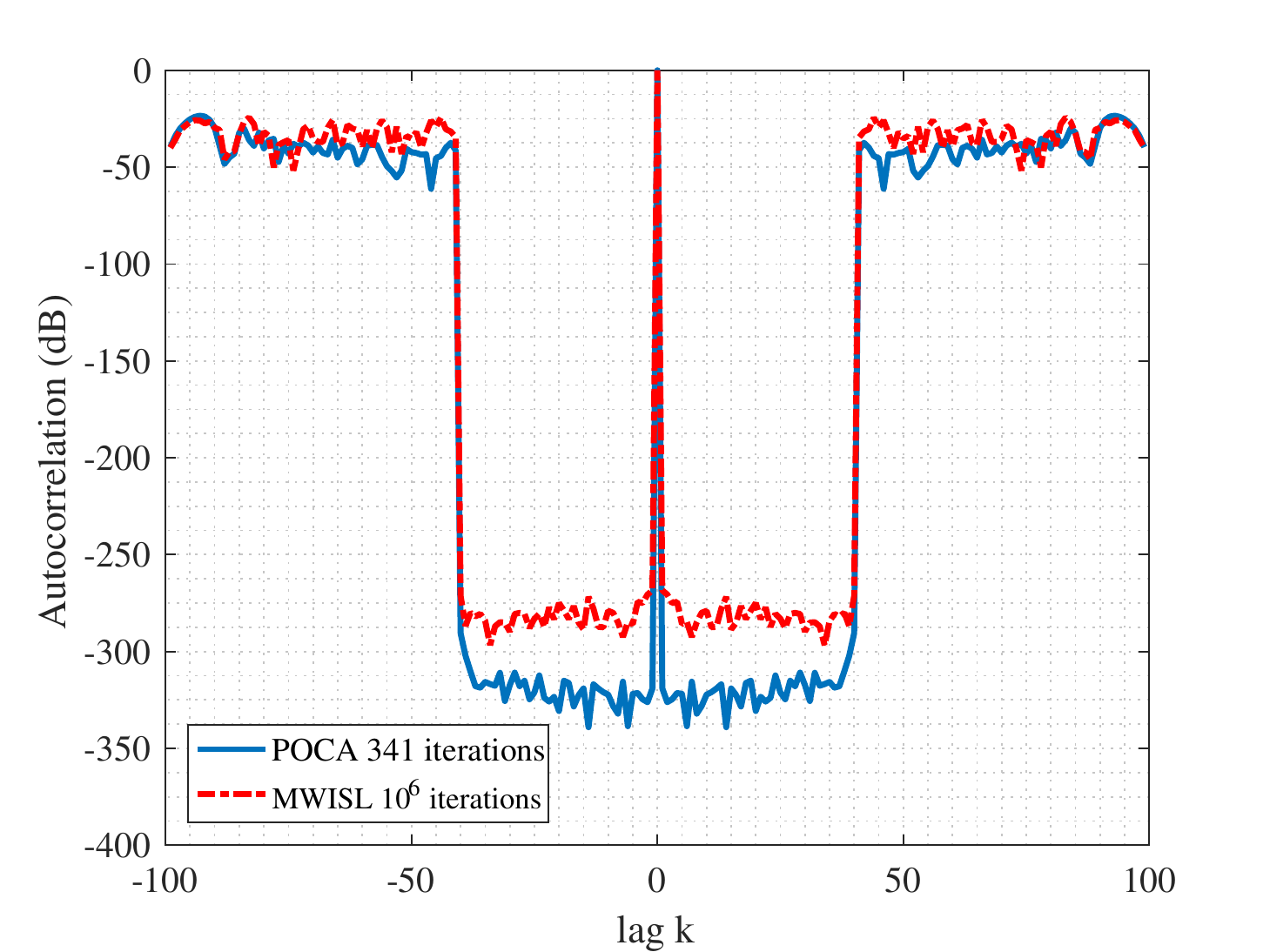}
		\vspace{-.5pc}\mvpc
		\caption{\small Normalized autocorrelation of POCA for $N=100,\,\,Q=39$(341 iterations), and MWISL \cite{Song2015} the same scenario ($10^6$ iterations)}
		\label{fig9}
		\vspace{-1.5pc}
	\end{figure}
	
	\begin{enumerate}[resume]
		\item Suppressing Less Than Half of the Side-lobes
	\end{enumerate}
	
	In the second scenario,  a situation is considered, where,  suppressing a specified part of the autocorrelation is  of concern. Consider the example given in equation (66) of \cite{Stoica2009}, where suppression of  $ r_{1} ,...,r_{39}$ is required when $N=100$, that is,
	\begin{equation} \label{ZEqnNum328910} \left\{\begin{array}{c} {Q=39\, \, \, } \\ { N=100  \, \, } \end{array}\right. .
	\end{equation}
	
	In \cite{Stoica2009}, the CAP+WeCAN is revealed to be the best-performing configuration (Refer to Fig. 4.b of \cite{Stoica2009}. Re-picturing is avoided to preserve brevity).
	The POCA's output initialized by modified Bernoulli is illustrated in \textbf{Fig.} \ref{fig9}  in companion with MWISL's output for the same scenario, where, the smallest correlation level and the peak correlation level in the suppressing region for POCA is -340 dB and -308 dB, respectively.
	In comparison, the corresponding figures for CAP+WeCAN are -320 and -280 dB, respectively, indicating an improvement of more than 20 dB. 
	In this figure the number of iterations for MWISL is set to $10^{6} $ to ensure convergence.
	In comparison, the required number of iterations for POCA algorithm in \textbf{Fig.}  \ref{fig9} is 341($\varepsilon =10^{-14} $). It is obvious that the supremacy of POCA to MWISL is more than 40 dB.
	
	The comparison of the MMF and MPCL figures for WeCAN initialized by CAP (CAP+WeCAN), POCA and RPOCA initialized by Modified Bernoulli (POCA+Modified Bernoulli and RPOCA+Modified Bernoulli) as well as the CAP algorithm are tabulated in Table \ref{table1}. In this table MPCL and MMF (Eq. \eqref{ZEqnNum569172},\eqref{ZEqnNum826115}) numbers are in magnitude. 
	From numbers in Table \ref{table1} it is deduced that, although POCA and RPOCA minimize PSL they show superior merit factors or equivalently ISL,
	that is, by decreasing peak of the side-lobes (or PSL), all the side-lobes and therefore their power integration (or ISL) would be decreased. This phenomenon is proved by authors for asymptotic case (i.e. for very large sequences) in \cite{H.E.Najafabadi2015}.
	The table provides a comparison between consumed time in seconds for each one of the aforementioned methods, as well. The RPOCA and POCA both consumes less than 3 sec, in comparison with more than 140 sec by the other methods. Note that the number for RPOCA is not less than POCA since the sequence is not long enough. In fact, the RPOCA is superior to POCA regarding time for large values of matrix dimension where the speedup of the fast SVD dominates time added by the additional steps in RPOCA.
	\begin{table}
		\centering 
		\caption{\small Comparison of MPCL and MMF for the scenario given in \eqref{ZEqnNum328910}}\label{table1}
		\begin{tabular}{|m{0.4in}|m{0.6in}|m{0.5in}|m{0.5in}|m{0.5in}|}  \hline
			& WeCAN +CAP & POCA +Modified Bernoulli & RPOCA +Modified  Bernoulli & CAP \\ \hline
			MPCL & $1.12\times 10^{-14}$ & $2.96\times10^{-15}$ & $3.096\times10^{-15}$ & $2.23\times10^{-13}$ \\ \hline
			MMF & $2.37\times10^{26}$ & $5.20\times10^{32}$ & $4.54\times10^{32}$ & $1.08\times10^{23}$ \\ \hline
			Consumed Time (s) & $ >140 $ & $ 2.36  $& $ 2.607 $ & $ 140.06 $ \\ \hline
		\end{tabular}
		\mvpc
	\end{table}

	\mvpc

	\begin{enumerate}[resume]\vpc
		\item Suppressing More Than Half of the Side-lobes
	\end{enumerate}
	
	As mentioned in \cite{Stoica2009}, the CAP and CAP+WeCAN algorithms are able to provide "almost zero" autocorrelation side-lobes "just" when suppressing less than half of the side-lobes are of concern. And if suppression of more than half is taken into account, then the side-lobes for either CAP or WeCAN becomes higher. According to the assessment made here, this holds for MWISL as well. 
	The reason behind this phenomenon is that the problem is formulated with unimodularity (a desirable but limiting) constraint. This additional constraint results in the elimination of the half of the degrees of freedom in nullifying autocorrelation values.
	Luckily, this limitation does not hold for PMQA, PMAR, POCA, and RPOCA, except when suppressing all the side-lobes is of concern, where the newly developed algorithms cannot provide almost zero side-lobes as well.
	Note that the scenario in \textbf{Fig.} \ref{fig5} affirms this fact for the first two algorithms.
	To illustrate this fact for POCA and RPOCA, several scenarios are contemplated here. 
	In the first, suppression of $r_{1} ,...,r_{64}$ at $N=100$ is envisioned and bearing in mind that typical surveillance pulse Doppler radars commonly utilize short or medium-size sequences for pulse compression, in the additional scenarios, the suppression of $\; r_{1} ,...,r_{32} $ at $N=40$ and suppression of $\; r_{1} ,...,r_{17} $ at $N=20$ are examined. \textbf{Fig.} 10 illustrates these scenarios.
	
	The question at hand now is that how much is the effect of chaotic sequence initialization on the autocorrelation side-lobes. 
	To answer this question, a 'suppressing more than half of the side-lobes' scenario is contemplated due to the fact that it has more practical value.
	Accordingly, the comparison for $ N=20,Q= 18 $ is illustrated in \textbf{Fig.} \ref{fig:pocavspocabernoulli}. 
	From this comparison, it is observed that the POCA initialized by modified Bernoulli performs slightly better (about 4 dB) for most of the suppression lags.
	\begin{figure}
		\centering
		\includegraphics[width=\linewidth,height=.7\linewidth]{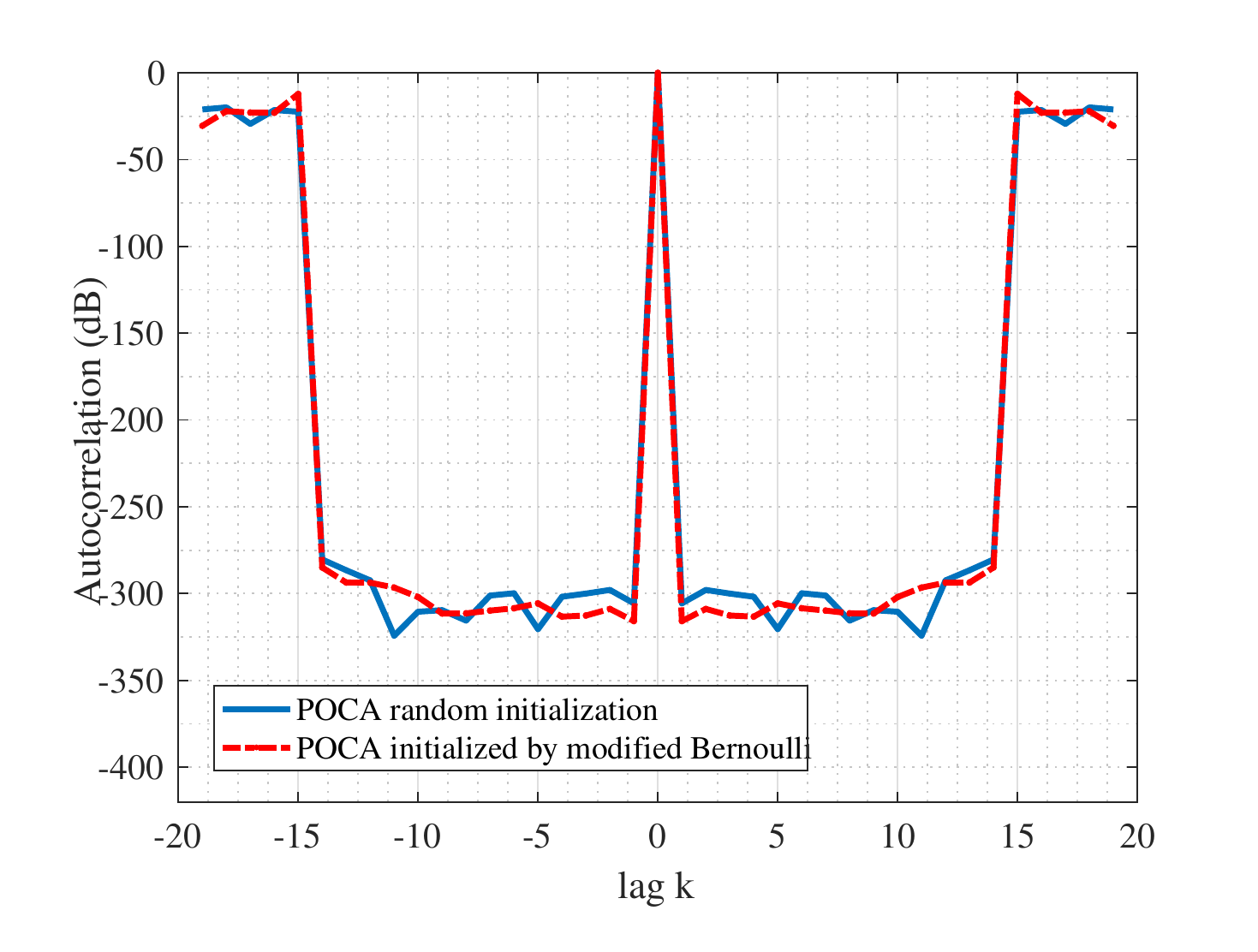}
		\mvpc\mvpc
		\caption{Comparison between POCA with random initialization and POCA initialized with modified Bernoulli}
		\label{fig:pocavspocabernoulli}\mvpc\hvpc
	\end{figure}
	\begin{figure}\mvpc
		\includegraphics[width=1\linewidth,height=.7\linewidth]{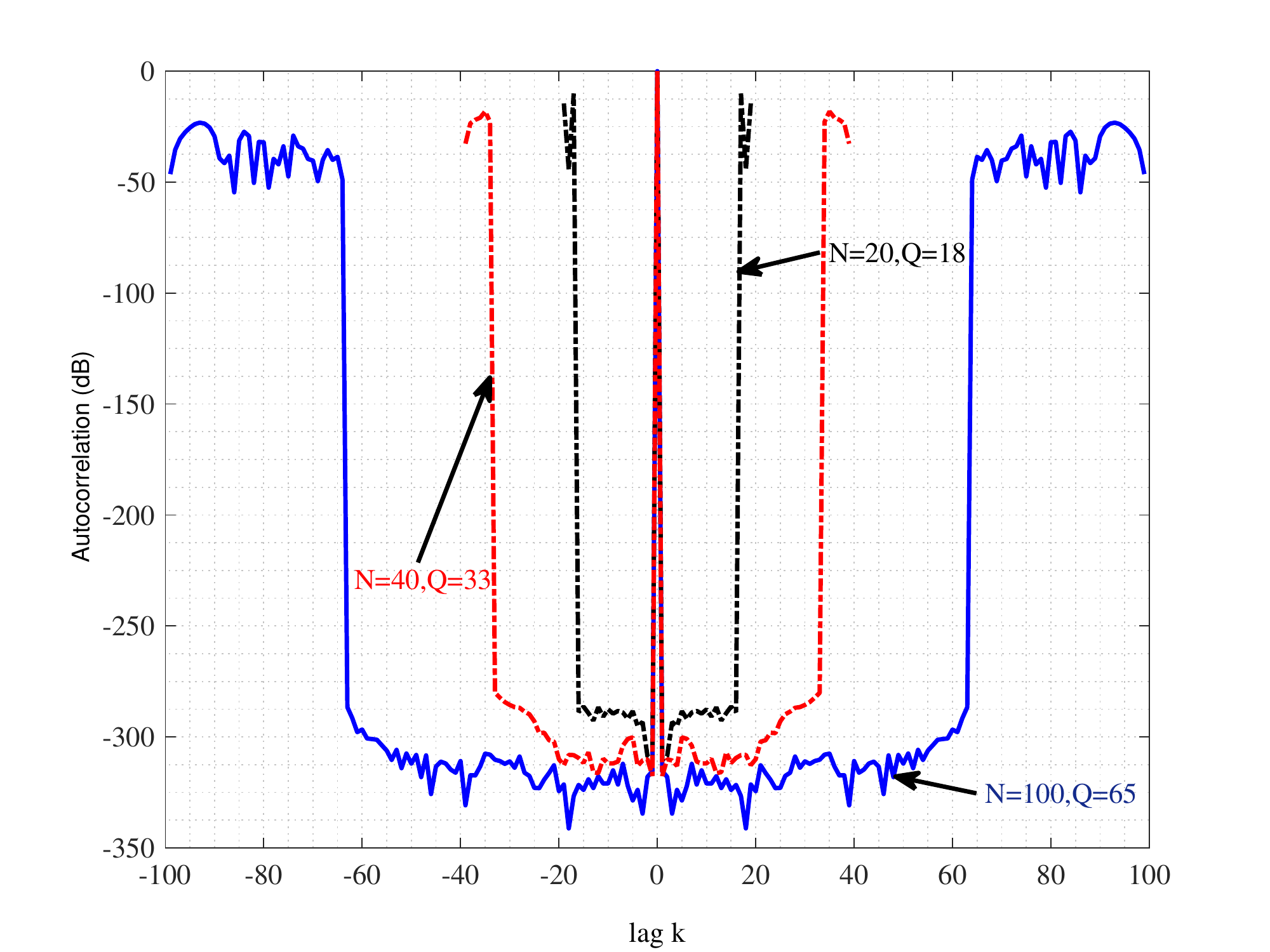}
		\label{fig:10}\mvpc
		\caption{\small Normalized Autocorrelation of POCA for $N=100,Q=65$; $N=40,Q=33$ and $N=20,Q=18$}
		\mvpc\hvpc
	\end{figure}
	
	\begin{enumerate}[resume]
		\item Dealing With Large $N$
	\end{enumerate}
	
	Another drawback of the CAP and WeCAN is that they cannot easily cope with large sequences (i.e. the lengths more than $N\sim 1000$). 
	In these situations, POCA and its fast version, RPOCA, can be applied. It should be noted that due to their iterative nature, the POCA and RPOCA have restrictions as well. 
	In fact, the superiority of these algorithms compared to their predecessors is of order one, two or three decades. 
	That is, POCA is restricted to lengths up to $N\sim 10^{4}$ and
	RPOCA cannot be applied for sequences more than $N\sim 10^{6}.$
	The application of this algorithm on Golomb sequence with lengths $N=10^{3}$ and $N=10^{4}$ are illustrated in \textbf{Fig} \ref{fig11a} and \ref{fig11b}, when suppression of $\; r_{1} ,...,r_{64}$ is required ($Q=65$). 
	The experiments here indicate that for RPOCA S as low as 4 would suffice.
	\begin{figure*}
		\centering
		\mvpc
		\begin{subfigure}{.95\columnwidth}
			\includegraphics[width=1\textwidth,height=.67\linewidth]{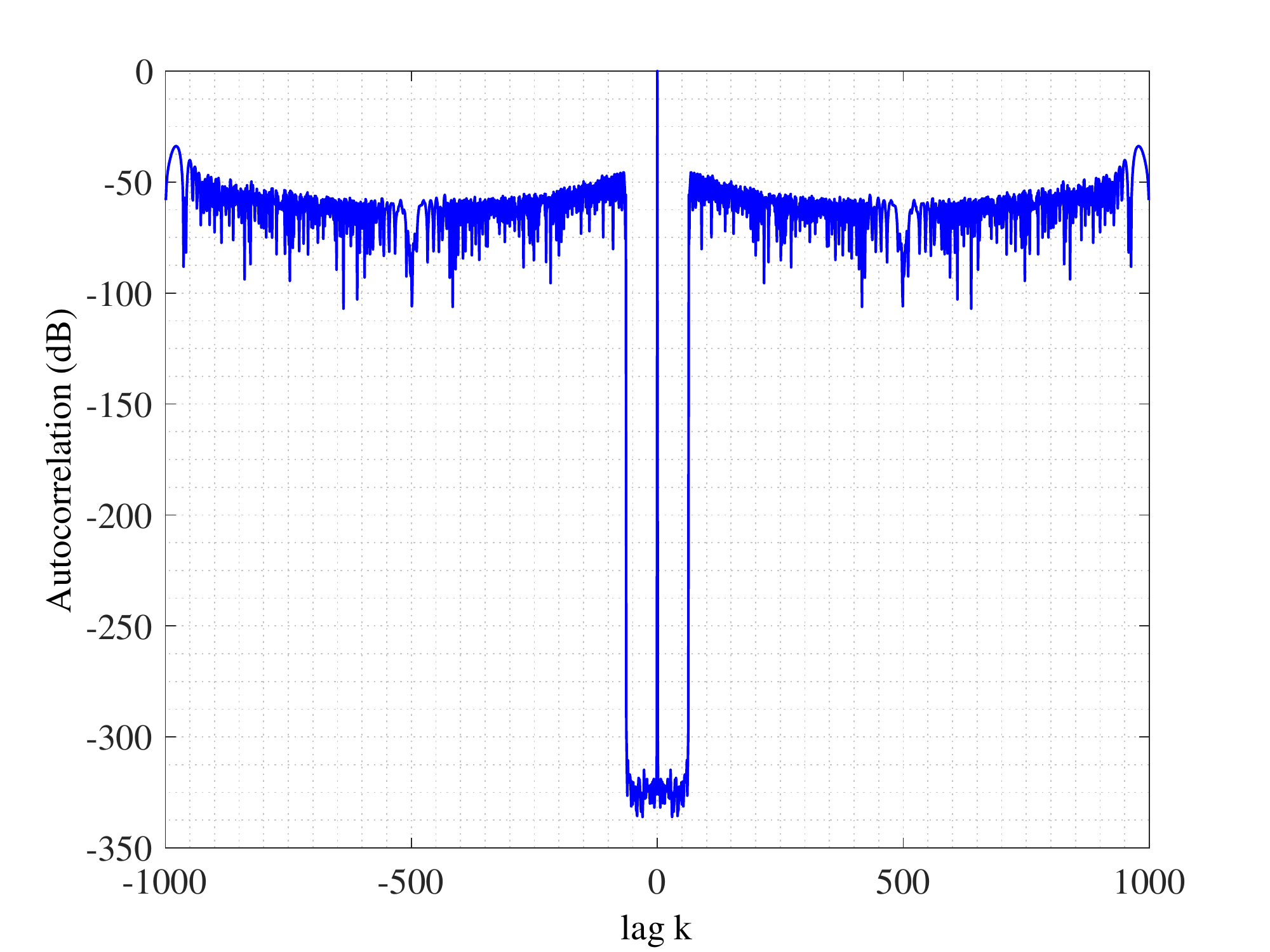}
			\caption{\small}\label{fig11a}
		\end{subfigure}
		\hfill
		\begin{subfigure}{.95\columnwidth}
			\includegraphics[width=1\textwidth,height=.67\linewidth]{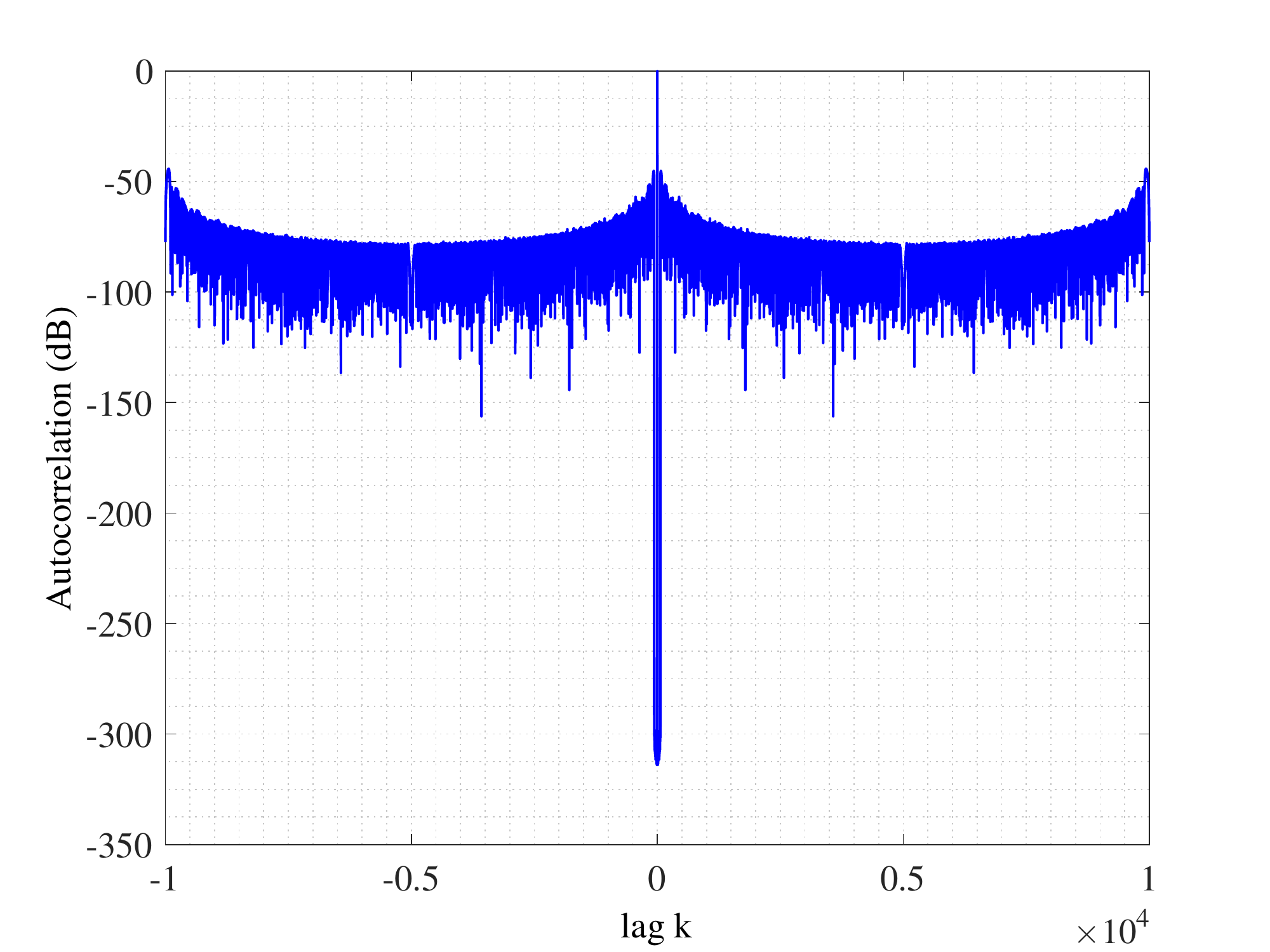}
			\caption{\small}\label{fig11b}
		\end{subfigure}
		\hfill
		\vspace{-.5pc}
		\label{fig11}
		\caption{\small Autocorrelation of POCA for $N=1000,Q=65$ and $N=10000,Q=65$}
		\vspace{-1.2pc}
	\end{figure*}
	
	\begin{enumerate}[resume]
		\item Imposing PAPR restriction and Unimodularity constraint
	\end{enumerate}
	Generally speaking, imposing additional restrictions like the ones introduced in section \ref{sec:Solv} projects time burden on the algorithm for total convergence.
	In another perspective, for a certain amount of iteration, PAPR restriction results in higher side-lobe levels.
	In order to show this fact, iteration numbers for POCA algorithm implemented with unimodularity and PAPR constraints are restricted. 
	In order to assess that the constrained version of the algorithm can suppress the side-lobes down to almost zero, the POCA constrained to unimodularity for $ N = 100, Q = 20 $ is simulated, where, the number of iterations is set to 10000 in order to assure convergence. 
	The result is depicted in \textbf{Fig.} \ref{fig:pocaconstrainedpapr} along with other scenarios. 
	Furthermore, in order to study the effect of PAPR bound, the POCA with constrained PAPR is examined, and to distinguish this case from others, the simulation parameters are $ N = 100, Q = 30 $. 
	Nevertheless, for a certain amount of iterations (1000 iterations)
	the POCA with constrained PAPR is simulated in three cases: unimodularity, $ a = 1.02 $, and $ a = 1.2 $, \textbf{Fig.} \ref{fig:pocaconstrainedpapr}, 
	where it is observed that less suppression (higher side-lobe levels) is reached when restriction on PAPR tightens.
	\vspace{-1.1pc}
	\begin{figure}\hvpc
		\raggedright
		\includegraphics[width=1.1\linewidth,height=.7\linewidth]{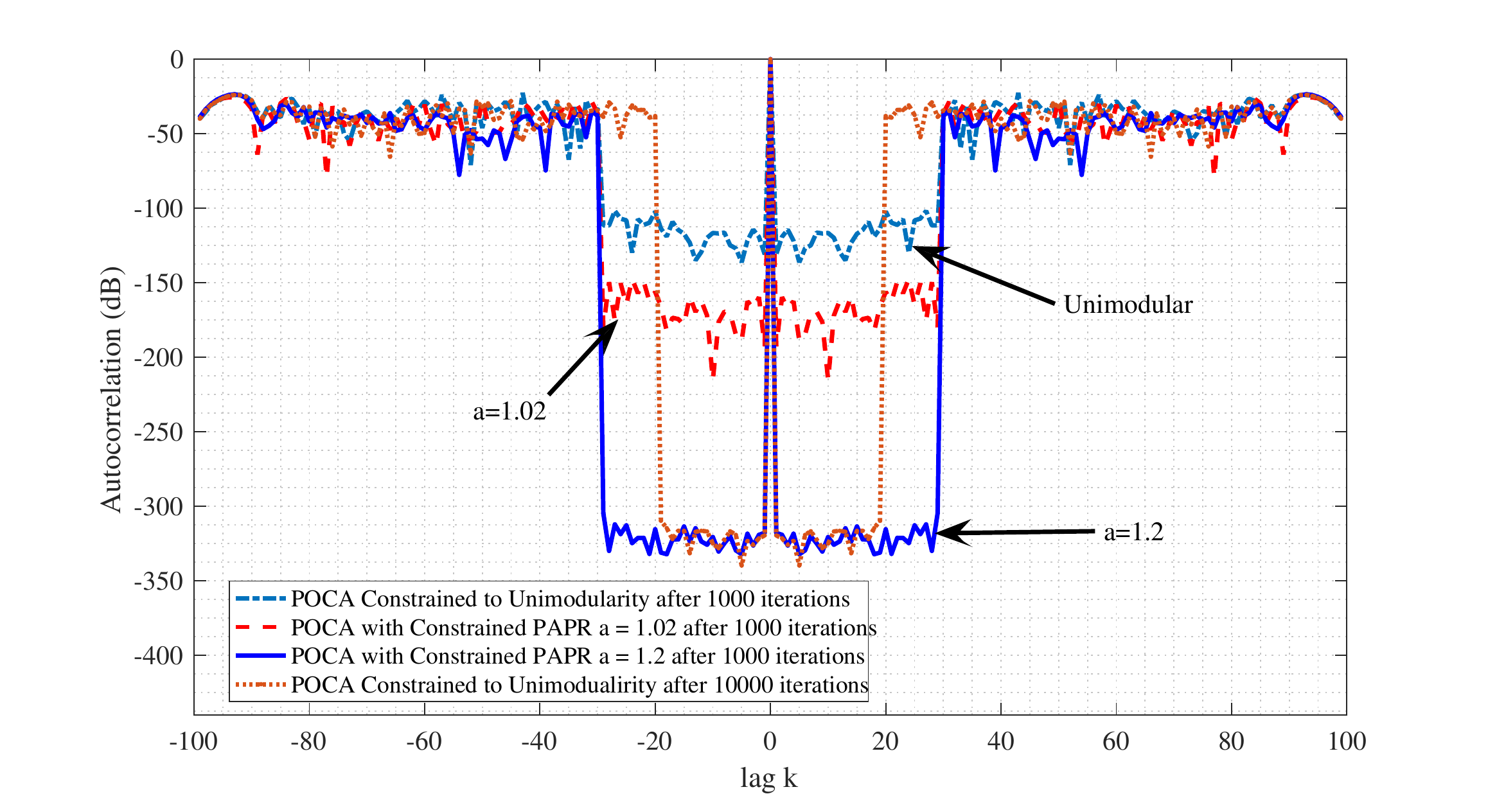}\hvpc
		\caption{Effect of imposing PAPR restriction on the suppression level for certain amount iterations}
		\label{fig:pocaconstrainedpapr}\mvpc\hvpc
	\end{figure}
	
	\subsection{Cross-Correlation Performance}
	\begin{table}
		
		\caption{Comparison of Computational burden when considering CCP}
		
		\begin{tabular}{|c|c|c|c|}
			\hline 
			\begin{tabular}{c|c}
				Consumed Time in seconds & \multirow{2}{*}{Length}\tabularnewline
				\cline{1-1} 
				Algorithm & \tabularnewline
			\end{tabular} & 10 & 100 & 300\tabularnewline
			\hline 
			\hline 
			POCA initialized by modified Bernoulli & 0.1 & 1.3 & 20.3 \tabularnewline
			\hline 
			MIMO-CAN initialized by modified Bernoulli & 0.1 & 3.2 & 83\tabularnewline
			\hline 
			MIMO-CAN random sequence initialization & 0.02 & 2.54 & 67.1 \tabularnewline
			\hline 
		\end{tabular}
		\hvpc\mvpc
	\end{table}    \hvpc
	From the above arguments, it is clear that the autocorrelation related metrics are improved. 
	However, the cross-correlation between the transmitted waveforms should be as low as possible in MIMO radars, 
	hence, to quantify the goodness of the cross-correlation, the cross-correlation peak (CCP) is defined according to,
	\begin{equation} \label{47)}
	CCP=\; {\rm max\; }\left(r_{xy} \left(k\right)\right),
	\end{equation}
	where, $r_{xy} \left(k\right)$ is the cross-correlation between the sequences $\{ x_{n} \} _{n=1}^{N}$ and $\{ y_{n} \} _{n=1}^{N}$,
	\begin{equation} \label{48)}
	r_{xy} (k)=\mathop{\sum }\limits_{n=k+1}^{N} x_{n} y_{n-k}^{*} =r_{xy}^{*} (-k)\, \, ,\, \, \, \, \, k=0,...,N-1.
	\end{equation}
	By this definition, the CCP goes to zero whenever all the cross-correlation values tends to zero.
	The comparison of cross-correlation peaks for Bernoulli system before and after modification is depicted in \textbf{Fig.} \ref{fig12a}. 
	This figure is generated by averaging CCP over 100 instances of the sequences generated by equations \eqref{ZEqnNum605756} and \eqref{ZEqnNum736077} initialized by different initial values. 
	The parameters of the systems in \eqref{ZEqnNum605756} and \eqref{ZEqnNum736077} are $\lambda =1.9$ and $B=1/\lambda$.
	Hence, the modification of the Bernoulli system results in the generation of a set of sequences with low cross-correlation.
	The comparison between the CCP of POCA + Modified Bernoulli is presented in \textbf{Fig.} \ref{fig12b}, when initialized by different initial values and MIMO-CAN algorithm \cite{He2009}, where the number of antenna's in MIMO-CAN is set to $M=40$.
	The MIMO-CAN is the generalization of CAN approach to MIMO radars by considering the cross-correlation in addition to autocorrelation.
	In \textbf{Fig.} \ref{fig12b}, the superiority of the POCA + Modified Bernoulli compared to MIMO-CAN is at least 12 dB. 
	Furthermore, by comparing \ref{fig12a} and \ref{fig12b} it is inferred that, when using POCA algorithm in order to suppress the autocorrelation of the modified Bernoulli, the cross-correlation degenerates (i.e. increases). 
	This fact indicates the existence of a trade-off of obtaining good autocorrelation and reduced cross-correlation.
	\begin{figure*}[!h]
		\centering\mvpc\mvpc
		\begin{subfigure}{\columnwidth}
			\includegraphics[width=0.95\textwidth,height=.7\linewidth]{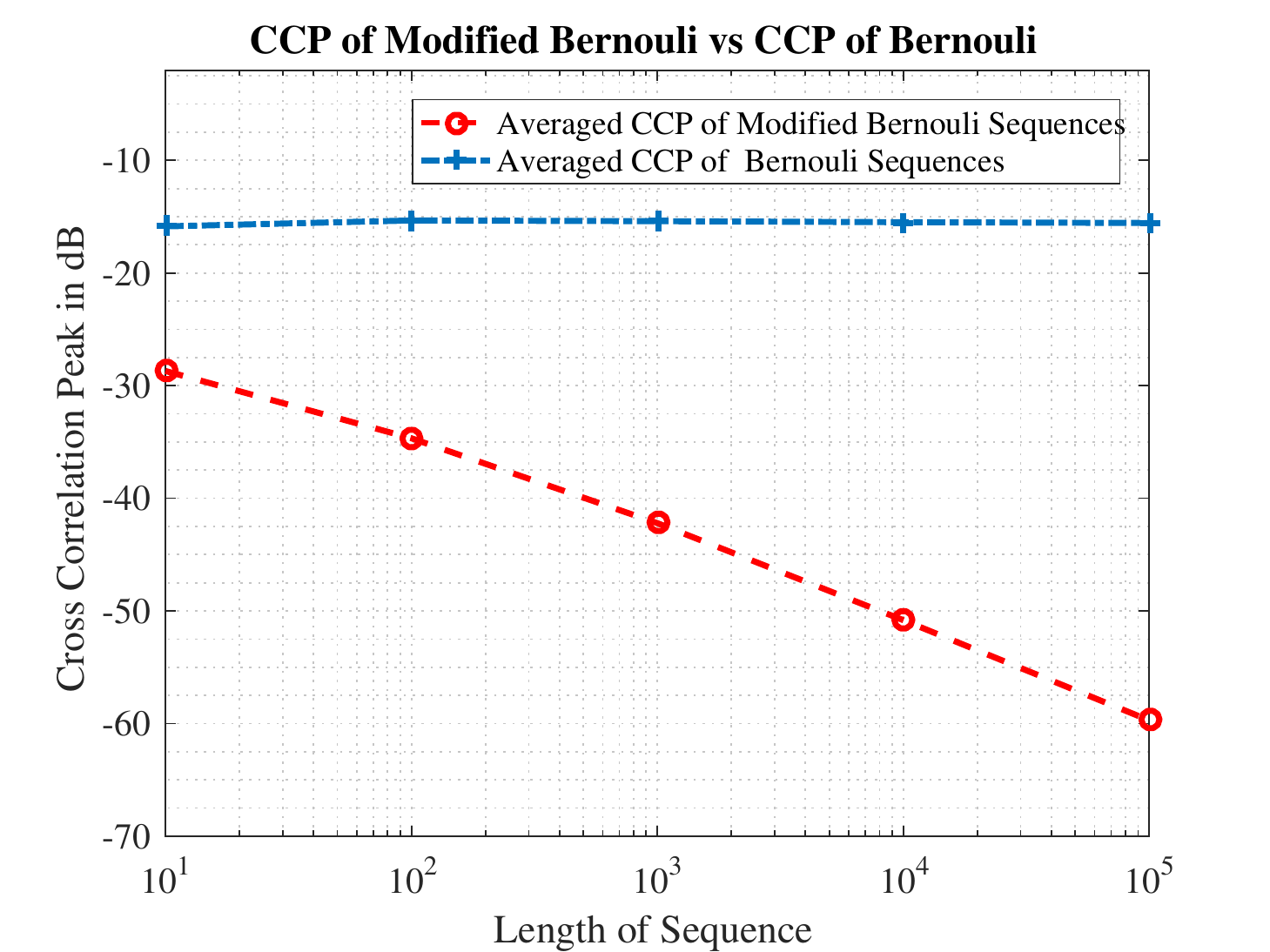}\hvpc
			\caption{\small}\label{fig12a}
		\end{subfigure}
		\hfill
		\begin{subfigure}{\columnwidth}
			\includegraphics[width=0.99\textwidth,height=.7\linewidth]{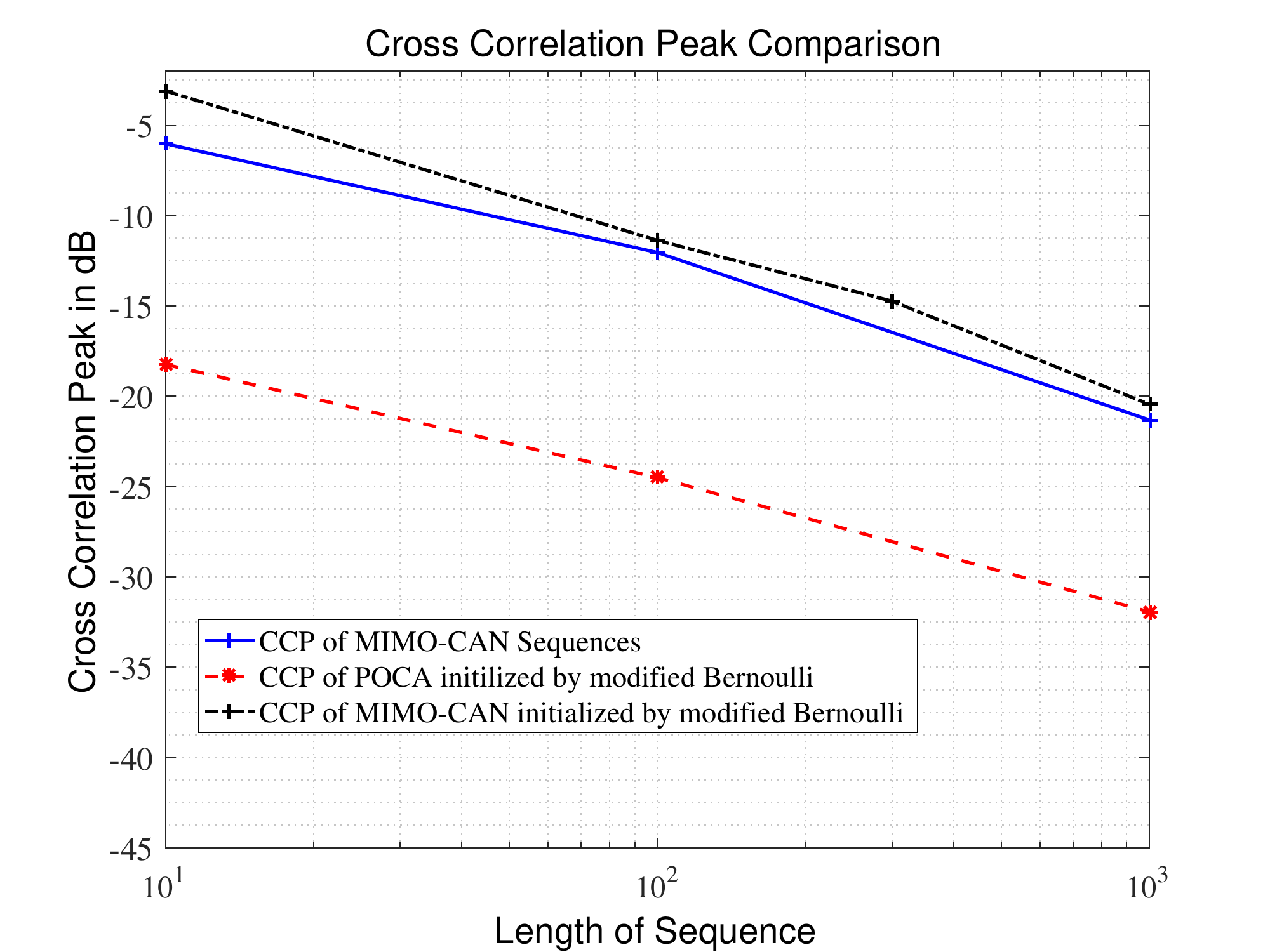}\hvpc
			\caption{\small}\label{fig12b}
		\end{subfigure}\hvpc
		\caption{\small Comparison of cross-correlation peak (CCP) between (a) Bernoulli and modified Bernoulli (b) MIMO-CAN with random sequences initialization vs MIMO-CAN + Modified Bernoulli  initialized by different initial values vs POCA + Modified Bernoulli initialized by different initial values.}
		\vspace{-1.2pc}\hvpc
	\end{figure*}
	\mvpc
	
	\subsection{Time performance}
	One of the advantages of the POCA algorithm over the monotonic minimizer is its low time consumption. 
	Therefore, the MWISL \cite{Song2015}, which is able to suppress the desired part of the autocorrelation is chosen as a benchmark and  
	the scenario in \eqref{ZEqnNum328910} is contemplated to compare consumed time and the MPCL.
	These two parameters are depicted in \textbf{Fig.} \ref{fig:timecompare} in an x-y plot, 
	where it is observed that for an equal amount of consumed time, the suppression level for POCA, measured by MPCL, is considerably more.
	
	\begin{figure}
		\centering
		\includegraphics[width=1\linewidth,height=.7\linewidth]{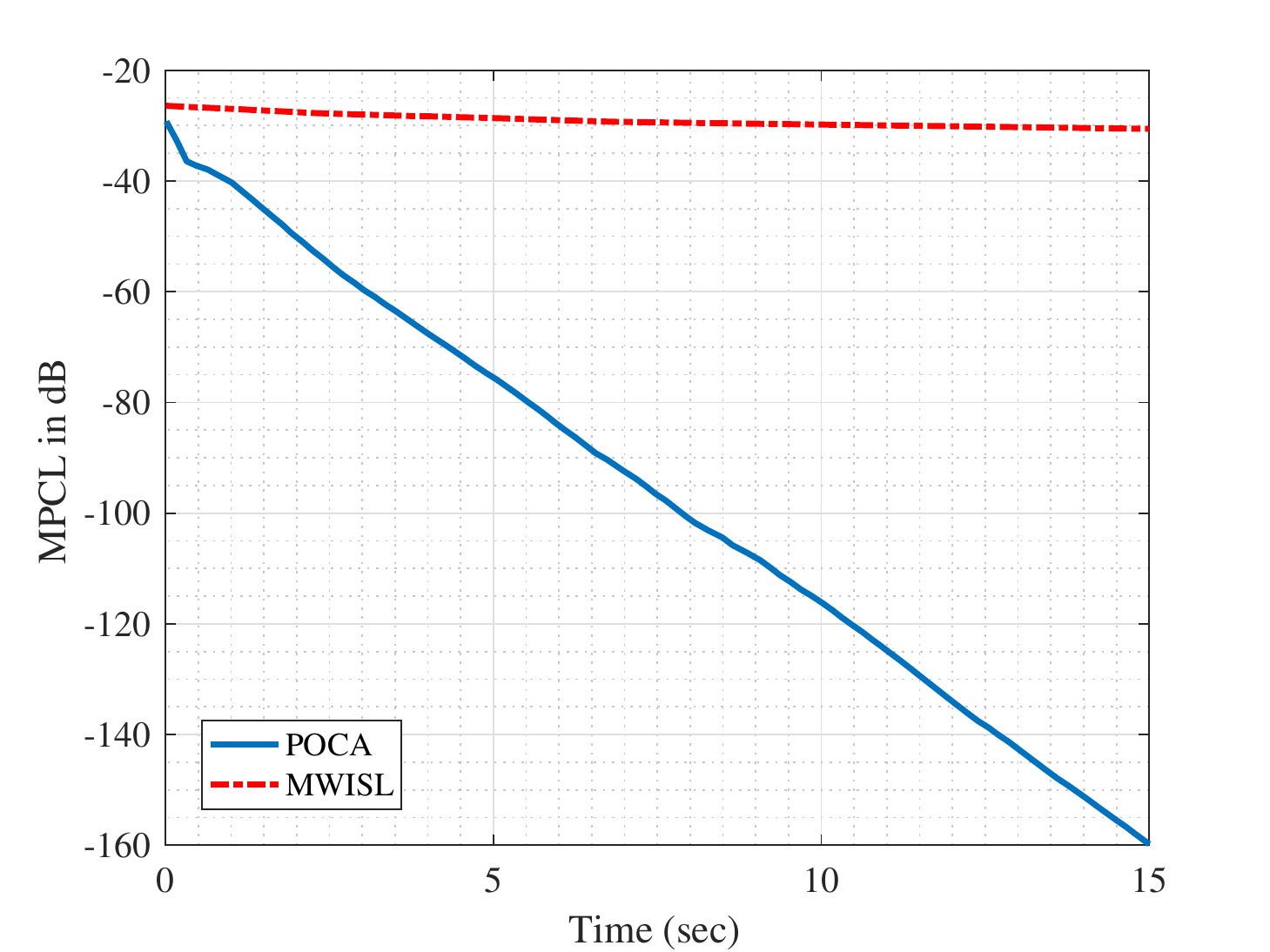}
		\caption{Comparison of the consumed time between POCA and MWISL for the scenario given in \eqref{ZEqnNum328910} .}\mvpc
		\label{fig:timecompare}
	\end{figure}
	
	\section{Conclusion}
	
	The  waveform generation for MIMO radars has been examined. 
	Our idea was that by juxtaposition of the chaotic systems and cyclic algorithms, it is possible to design large set of sequences which,
	
	\begin{itemize}
		\item  have very low autocorrelation side-lobes.
		
		\item  can have ``almost zero" side-lobes in specified regions of autocorrelation.
		
		\item  maintain low cross-correlation between the elements of the sets.
	\end{itemize}
	In the way of obtaining the aforementioned set,  the problem of minimizing peak side-lobe level is formulated. 
	Consequently, several algorithms are developed namely, PMQA, PMAR, POCA, and RPOCA. 
	The self-applied algorithms to mono-static pulse Doppler (PD) radars outperform their predecessors both in the terms of implementation time 
	and their generated side-lobe level of the sequences. 
	Moreover, the Bernoulli chaotic system is modified to have low cross-correlations. 
	By juxtaposition of the previously mentioned algorithms and modified Bernoulli, appropriate waveforms for MIMO radar applications can be generated.\mvpc
	
	\section{Appendix A} \label{sec:appA}

	Throughout the paper, the notion that $( {\mathbb C}^{N\times M} ,\tbd.\tbd _{\infty })$ is a Banach space has been used extensively. In order to maintain the entirety of the article, the credibility of this fact should be verified. Consider the vector space of all $N\times M$ matrices over the field of complex numbers, or ${\mathbb C}^{N\times M} .$ Let $\bA$ and $\bB$ be two arbitrary elements of this vector space. Moreover, let $\alpha $ be an arbitrary complex number. In order to maintain that $({\mathbb C}^{N\times M} ,\tbd.\tbd_{\infty })$ is a normed space, it is sufficient to prove that,
	
	\begin{itemize}
		\item  Positive definiteness:
		\begin{equation} \label{ZEqnNum798489} \begin{array}{l} {\tbd \bA \tbd_{\infty } >0\, \, {\rm for\; all\; nonzero\; vectors\; and}\, } \\ {\tbd \bA \tbd_{\infty } =\b0\, \, {\rm if\; and\; only\; if}\, \bA=\b0} \end{array} \end{equation}
		
		\item  Linearity:
		\begin{equation} \label{ZEqnNum995007} \tbd\alpha  \bA \tbd_{\infty } =\alpha \tbd \bA \tbd_{\infty }  \end{equation}
		
		\item  Triangle inequity:
		\begin{equation} \label{51)} \tbd \bA+\bB \tbd_{\infty } \, \le \, \, \, \tbd \bA \tbd_{\infty } +\tbd \bB \tbd_{\infty } .
		\end{equation}
	\end{itemize}
	
	The two first properties, \eqref{ZEqnNum798489} and \eqref{ZEqnNum995007} are straightforward. To prove that last note that
	\begin{equation} \label{ZEqnNum481496}
	\begin{array}{l}
	{\tbd\bA+\bB\tbd_{\infty } \, =\max \left|A_{i,j} +B_{i,j} \right|\, } \\
	{\le \max (\left|A_{i,j} \right|\, \, +\left|B_{i,j} \right|\, )} 
	{=\max \, \left|A_{i,j} \right|\, +\max \left|B_{i,j} \right|} \\
	{=\, \tbd\bA\tbd_{\infty } + \tbd\bB\tbd_{\infty } ,}
	\end{array}
	\end{equation}
	where, triangle inequity for complex numbers is used in forming \eqref{ZEqnNum481496}. For a normed space to be Banach space an additional condition known as ``completeness'' is essential. In specific terms, given any Cauchy series in the space, it should converge.
	The completeness of $({\mathbb C}^{N\times M} ,\tbd.\tbd_{\infty })$ is revealed here through Theorem \ref{th2}. Before that, the term ``Cauchy series'' is defined,
	
	\textbf{Definition:}
	\textit{Consider a normed vector space }$( \bX,\tbd.\tbd)$\textit{, a series }$\left\{x_{n} \right\}_{n=1}^{\infty } $\textit{ with its elements in }$\bX$ \textit{is a Cauchy series if,}
	\begin{equation} \label{54)}
	\, \forall \varepsilon >0\, \, \, \, \, \exists N_{0} \, \, ;\, \, \, \, \forall m,n>N_{0} \, \, \, \, \, \left|x_{n} -x_{m} \right|<\varepsilon .
	\end{equation}
	
	\begin{theorem}\label{th2}
		
		The metric space defined by the $\tbd.\tbd_{\infty }$ norm
		over ${\mathbb C}^{N\times M} $ is complete.
	\end{theorem}
	\begin{proof}
		
		Let, $\{ A^{n} \} _{n=0}^{\infty } $ be a Cauchy series. Note that in this notation $n$ is just a superscript and does not denote the power. For any $\varepsilon >0$, there exists an integer positive $N_{0} $ such that for every $n,m>N_{0} $,
		\begin{equation} \label{ZEqnNum652987} \tbd\bA^{m} -\bA^{n} \tbd_{\infty } \, \le \, \, \varepsilon . \end{equation}
		In order to show that such a sequence converges, it is sufficient to prove that there exists a matrix like $\bA$ such that for any $\varepsilon >0$, there is an integer $N_{0} >0$ such that
		\begin{equation} \label{56)}
		\tbd\bA^{n} -\bA\tbd_{\infty } \, \le \, \, \varepsilon .
		\end{equation}
		From \eqref{ZEqnNum652987}, it is revealed that
		\begin{equation} \label{57)}
		\max \left|A_{i,j}^{n} -A_{i,j}^{m} \right|<\varepsilon
		\end{equation}
		Therefore, for any $\varepsilon >0$, there exists a positive integer, $N$ such that for any $n,m>N$ and for every $i$ and $j$ the inequality $ \left|A_{i,j}^{n} -A_{i,j}^{m} \right|<\varepsilon  $ is guaranteed. Alternatively, for every $i$ and $j$ the complex sequence generated by $\left\{A_{_{i,j} }^{n} \right\}_{n=0}^{\infty } $ is a Cauchy series in $({\mathbb C},|.|)$ normed vector space, where $|.|$ denotes the absolute value. Note that, $(\; {\mathbb C},|.|)$ is known to be a complete normed vector space, indicating that, every Cauchy series in it converges. For the sequences  $\left\{A_{i,j}^{n} \right\}_{n=0}^{\infty } $ assume, they converge to $f_{i,j}^{} $ for every $i$ and \textit{j.} Mathematically speaking, 
		\begin{equation} \label{ZEqnNum528600} \, \forall \varepsilon >0\, \, \, \, \, \, \, \, \exists N_{0} \, \, ;\, \, \, \, \forall n>N_{0} \, \, \, \, \forall _{i,j} \, \, \, \left|A_{i,j}^{n} -f_{i,j}^{} \right|<\varepsilon .    \end{equation}
		Note that, the inequality holds for every $i$ and $j$. Therefore, it also holds when ``every" is changed to ``max". That is, the statement in \eqref{ZEqnNum528600} is equivalent to,
		\begin{equation} \label{60)} \, \forall \varepsilon >0\, \, \, \, \, \, \, \, \exists N_{0} \, \, ;\, \, \, \, \forall n>N_{0} \, \, \, \, \max \, \left|A_{i,j}^{n} -f_{i,j}^{} \right|<\varepsilon .   \end{equation}
		Hence,
		\begin{equation} \label{ZEqnNum176445} \, \forall \varepsilon >0\, \, \, \, \, \, \, \, \exists N_{0} \, \, ;\, \, \, \, \forall n>N_{0} \, \, \, \, \tbd \bA^{n} -\bF\tbd_{\infty } \, \le \, \, \varepsilon . \end{equation}
		where, $\bF$ is defined according to the following equation.
		\begin{equation}\label{key}
		{\bf{F}} = \left[ {{f_{i,j}}} \right],\,\,\,i = 1,...,N,\,\,\,j = 1,...,M.
		\end{equation}
		The statement in \eqref{ZEqnNum176445} states that the sequence $\{ \bA^{n} \} _{n=0}^{\infty } $ converges to $\bF$ in the Chebyshev distance sense, resulting in the completeness of the proof.
	\end{proof}

	\mvpc\hvpc
	\section{Appendix B}\label{sec:appB}
	
	In this appendix alternative approaches to exact solution of the smallest circle problem in \eqref{ZEqnNum175528} is developed. For any complex sequence $\left\{\mu _{k} \right\}_{k=1}^{Q-1} $ consider the following problem,
	\begin{equation}\label{ZEqnNum491174}
	\begin{gathered}
	\mathop{\min }\limits_{x} \mathop{\max }\limits_{k} \left|x-\mu _{k} \right| \\
	{=\mathop{\min }\limits_{x_{R} ,\, x_{I} } \left\{\mathop{\max }\limits_{k} \sqrt{\left(x_{R} -Re\left\{\mu _{k} \right\}\right)^{2} +\left(x_{I} -Im\left\{\mu _{k} \right\}\right)^{2} } \right\}.}
	\end{gathered}
	\end{equation}
	First, notice that this unconstrained problem is equivalent to the following constrained problem,
	\begin{equation}\label{ZEqnNum177560}
	\begin{gathered}
	\mathop{\min }\limits_{x_{R} ,\, x_{I} } R \\
	s.\, t.\, \, \sqrt{\left(x_{R} -Re\left\{\mu _{k} \right\}\right)^{2} +\left(x_{I} -Im\left\{\mu _{k} \right\}\right)^{2} } \le R,
	\end{gathered}
	\end{equation}
	where, $x_{R} $ and $x_{I} $ are real and imaginary parts of $x$.
	Then, by defining $y_{k} :=x_{I} -Im\left\{\mu _{k} \right\}$, $w_{k} :=x_{R} -Re\left\{\mu _{k} \right\}$, $k=1,\ldots , Q-1$ this problem can be restated as,
	\begin{equation}\label{ZEqnNum733358}
	\begin{gathered}
	{\mathop{\min }\limits_{x_{R} ,\, x_{I} ,\left\{w_{k} \right\},\left\{y_{k} \right\}} R} \\
	{s.\, t.\, \, \, \, w_{k} +x_{R} =Re\left\{\mu _{k} \right\},\, \, k=1,\; \ldots ,\; Q-1} \\
	{\, \, \, \, \, \, \, \, \, \, y_{k} \, \, +x_{I} =Im\left\{\mu _{k} \right\}\; ,\, \, k=1,\; \ldots ,\; Q-1} \\
	{\, \, \, \, \, \, \, \, \, R\ge \sqrt{w_{k}^{2} +y_{k}^{2} } \; \, \, \, \, \, \, \, \, \, \, ,\, \, k=1,\; \ldots ,\; Q-1.}
	\end{gathered}
	\end{equation}
	The problem in \eqref{ZEqnNum733358} is a cone optimization problem and can be solved through interior point method with an arithmetic cost of $O(n^{3.5} \left|{\rm \; log\; }\delta \right|)$, where $\delta $ is the user defined accuracy parameter.
	
	Another approach to solving \eqref{ZEqnNum491174} is obtained through noticing that this problem is also an unconstrained nondifferentiable convex programming, solvable through  subgradient method. Accordingly, one may follow the following steps in order to find the exact solution.
	
	\begin{enumerate}
		\item  Start form an appropriate point, for instance
		\begin{equation}\label{e76}
		x^{0} =\; \frac{\max ^{d} \left\{\mu _{k} \right\}+\min ^{d} \left\{\mu _{k} \right\}}{2} .
		\end{equation}

		\item  Compute the subgradient of
		\begin{equation}\label{e77}
		\begin{gathered}
		f(x_{R} ,x_{I} )= \\
		\mathop{\max }\limits_{k} \sqrt{\left(x_{R} -Re\left\{\mu _{k} \right\}\right)^{2} +\left(x_{I} -Im\left\{\mu _{k} \right\}\right)^{2} },
		\end{gathered}
		\end{equation}

		at $(x_{R}^{i} ,x_{I}^{i} )$
		
		\item  Determine the step size by using line search. To this end, let
		\begin{equation}\label{66)}
		\begin{gathered}
		K=\{ k:f(x_{R} ,x_{I} )=  \\
		\sqrt{(x_{R} -Re\{ \mu _{k} \} )^{2} +(x_{I} -Im\{ \mu _{k} \} )^{2} } \} ,
		\end{gathered}
		\end{equation}
		be the active functions index set. Then, constitute the set,
		\begin{equation} \label{67)}
		\begin{array}{l} {\partial f(x_{R} ,x_{I} )=} \\ {co\left\{\frac{\left[\begin{array}{c} {x_{R} -Re\left\{\mu _{k} \right\}} \\ {x_{I} -Im\left\{\mu _{k} \right\}} \end{array}\right]}{\sqrt{\left(x_{R} -Re\left\{\mu _{k} \right\}\right)^{2} +\left(x_{I} -Im\left\{\mu _{k} \right\}\right)^{2} } } |k\in K\right\},} \end{array}
		\end{equation}
		where "co" and $\partial f(x_{R} ,x_{I} )$ denote the convex hull and the sub-differential of $f(x_{R} ,x_{I} )$, respectively. Therefore, any member of this set is a sub-gradient of $f$. Accordingly, if $\partial f(x_{R} ,x_{I} )$ has just one member, it means $f$ is differentiable and the steepest decent direction (i.e. the negative of the sole member of $\partial f(x_{R} ,x_{I})$) should be taken as search direction vector. Otherwise, the negative sub-gradient with the smallest norm should be taken.

		\item  Set $i=i+1$ and go to step 2 unless the step size is less than a specified threshold.
	\end{enumerate}
	
	It should be noted that there may be other solutions to the smallest enclosing circle problem, which are beyond the topic of this paper. Moreover, like the solutions given in this appendix all the solutions for smallest circle problem involve iteration, and thus utilizing them yields to approximate solution after some iterations. This fact is the reason behind finding a heuristic approximate solution in \eqref{ZEqnNum457011}. 
	Note that, if the problem in \eqref{ZEqnNum175528} is confined to real sequences, the solution in \eqref{ZEqnNum457011} is exact solution.
	This fact is revealed in the following Lemma.
	\hvpc
	\begin{lemma}
		The smallest circle problem for real-valued sequence has a solution of the form,
		\begin{equation} \label{68)} x_{*} =\; \frac{\max \mu _{k} +\min \mu _{k} }{2}  \end{equation}
	\end{lemma}
	
	\begin{proof}
		First, notice that the feasible set is not empty. Divide the feasible set into two following subsets:
		
		\begin{enumerate}
			\item  The elements of feasible set that are greater than $x_{*} $. That is,
			
			\begin{equation} \label{ZEqnNum106679} x\ge \frac{\max \mu _{k} +\min \mu _{k} }{2}  \end{equation}
			This inequality implies that,
			\begin{equation} \label{ZEqnNum889251}
			x-\; \min \mu _{k} \ge \left|x-\mu _{k} \right|,\; \; \forall k\;
			\end{equation}
			where the equality holds for those $k$ where minimum of $\mu _{k} $ occurs. From \eqref{ZEqnNum889251} it is easy to deduce that,
			
			\begin{equation} \label{71)} \mathop{\max }\limits_{k} \left|x-\mu _{k} \right|=x-\; \min \mu _{k}  \end{equation}
			So the problem in \eqref{ZEqnNum175528}  is equal to the minimization problem:
			\begin{equation} \label{ZEqnNum807248} \begin{array}{c} {\mathop{\min }\limits_{x} \left(x-\; \min \mu _{k} \right)} \\ {s.t.\; \, \, \, x\ge \frac{\max \mu _{k} +\min \mu _{k} }{2} \; } \end{array} \end{equation}
			Obviously, the solution to this problem is:
			\begin{equation} \label{73)} x_{*} =\; \frac{\max \mu _{k} +\min \mu _{k} }{2} .            \end{equation}
			\item  The elements of feasible set that are less than $x_{*} $, that is,
			\begin{equation} \label{74)} x\le \frac{\max \mu _{k} +\min \mu _{k} }{2} .        \end{equation}
			Similarly, here we have:
			\begin{equation} \label{75)} {\rm max\; }\mu _{k} -\; x\ge \left|x-\mu _{k} \right|,\; \; \forall k\;  \end{equation}
			or,
			\begin{equation} \label{76)} \mathop{\max }\limits_{k} \left|x-\mu _{k} \right|=max\; \mu _{k} -\; x.      \end{equation}
			Thus, in this case the minimization problem in \eqref{ZEqnNum175528} is equal to the following problem:
			\begin{equation} \label{77)} \begin{array}{c} {\mathop{\min }\limits_{x} (max\; \mu _{k} -\; x)} \\ {s.t.\, \; x\le \frac{\max \mu _{k} +\min \mu _{k} }{2} \; } \end{array}.
			\end{equation}
			With the change of variable $z=-x$, this problem can be restated as:
			\begin{equation} \label{78)} \begin{array}{c} {\mathop{\min }\limits_{z} (z+\; max\; \mu _{k} )} \\ {s.t.\; \, \, z\ge -\frac{\max \mu _{k} +\min \mu _{k} }{2} \; } \end{array},
			\end{equation}
			which is the same problem as in  \eqref{ZEqnNum807248} with,
			\begin{equation}\label{79)}
			z_{*} =-\; \frac{\max \mu _{k} +\min \mu _{k} }{2}
			\end{equation}
		\end{enumerate}
		
		Finally, note that the inclusion of the $x_{*} $ in both subsets does not alter the integrity of the proof.
	\end{proof}
	\hvpc
	\bibliographystyle{plain}
	\bibliography{IEEEabrv,2PSL}
	\vspace{-2.5pc}
	\begin{IEEEbiography}[{\includegraphics[bb=0mm 0mm 208mm 296mm, width=25mm, height=31.6mm, viewport=3mm 4mm 205mm 292mm]{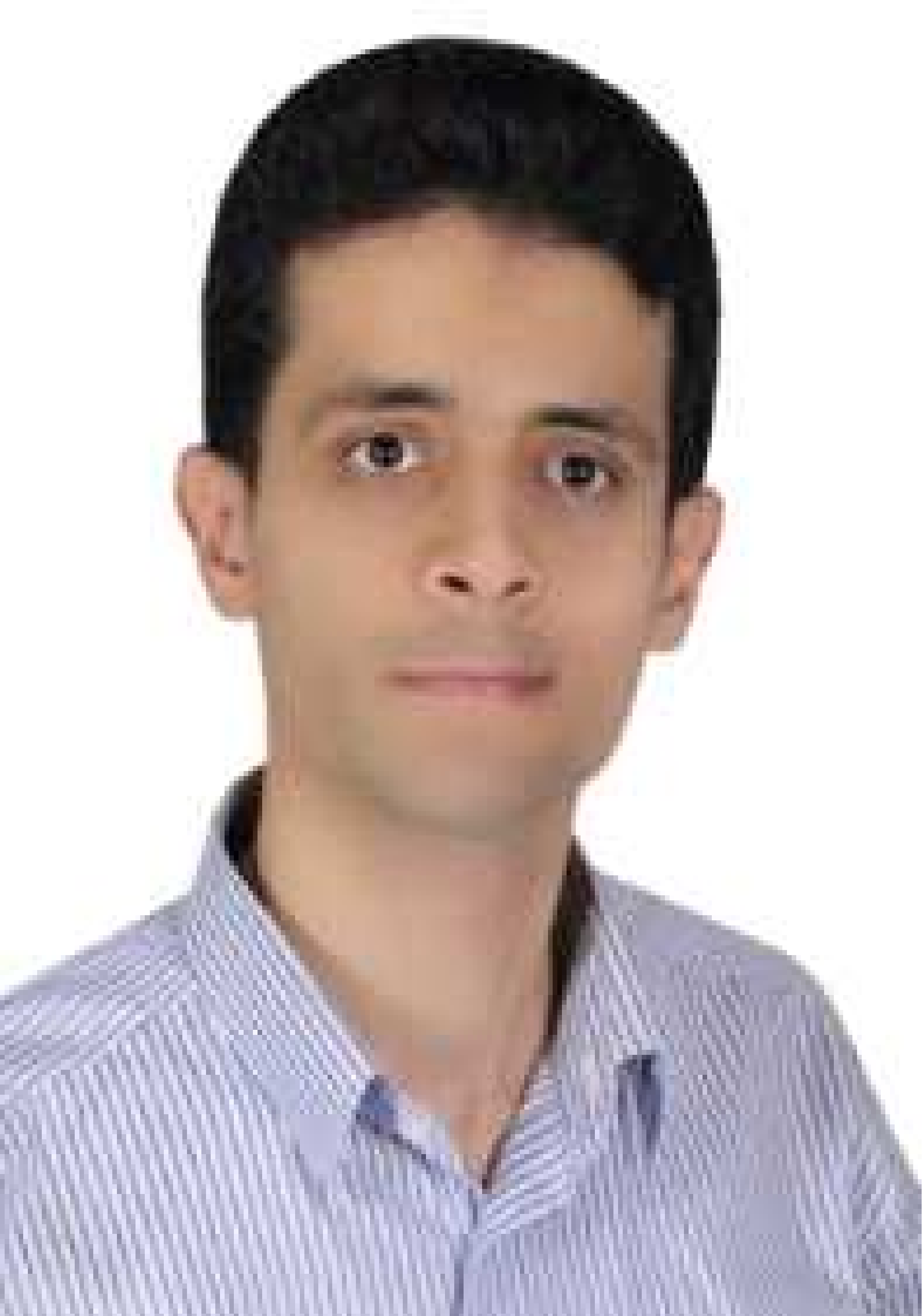}}]
		{H. E. Najafabadi received B.S. and M.S. in electronics, and communication in 2007 and 2009 from Isfahan University of Technology and Amirkabir University of Technology, respectively. From 2009 to 2010, he was with the ICTI group (icti.ir), where he was working on the national airborne radar system as a team member. Currently, He works toward Ph.D. in University of Isfahan.}
	\end{IEEEbiography}
	\vspace{-3pc}
	\begin{IEEEbiography}
		[{\includegraphics[bb=0mm 0mm 208mm 296mm, width=25mm, height=31.6mm, viewport=3mm 4mm 205mm 292mm]{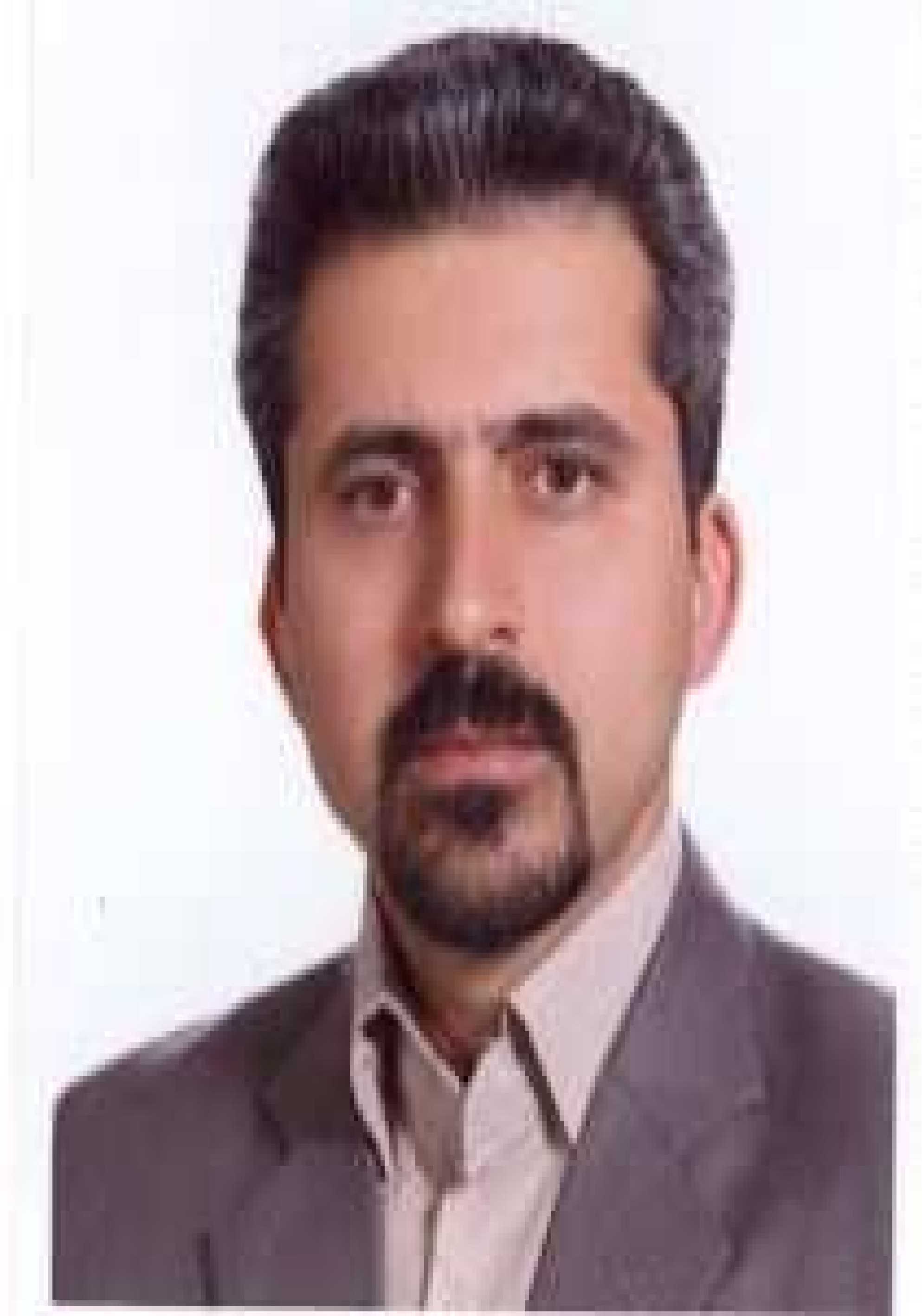}}]
		{\scriptsize M. Ataei received the B.Sc. degree from the Isfahan University of Technology, Iran, in 1994, the M. Sc. degree from the Iran University of Science \& Technology, Iran, in 1997, and PhD degree from K. N. University of Technology, Iran, in 2004 (joint project with the University of Bremen, Germany) all in control Engineering. He is an associate professor at the Department of Electrical Eng. at University of Isfahan, Iran. His main areas of research interest are control theory and applications, nonlinear control, and chaotic systems' analysis and control.}
	\end{IEEEbiography}
	\vspace{-3.5pc}
	\begin{IEEEbiography}
		[{\includegraphics[bb=0mm 0mm 208mm 296mm, width=25mm, height=31.6mm, viewport=3mm 4mm 205mm 292mm]{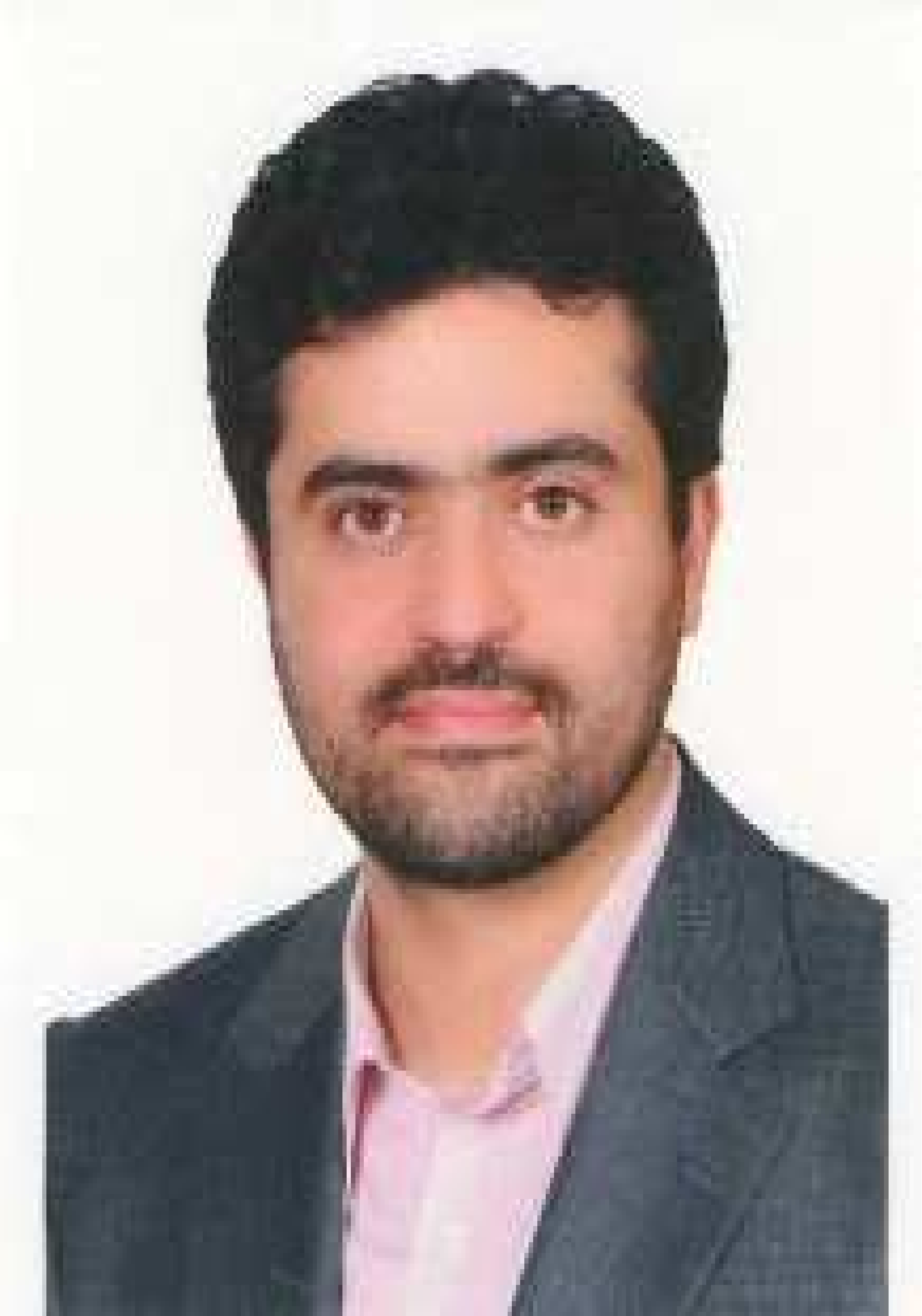}}]
		{\scriptsize M. F. Sabahi was born in Isfahan in 1976. He received B.S. and M.S. degrees, in 1998 and 2000, in Electronics Engineering and Communication Engineering, respectively, from Isfahan University of Technology; Isfahan, Iran. He also received the Ph.D. in Electrical Engineering from the same university in the year 2008. He has been a faculty member of Electrical Engineering Department, at the University of Isfahan from 2008 till now. His main research interests include Statistical Signal Processing, Detection Theory, and Wireless Communication.}
	\end{IEEEbiography}

\end{document}